\documentclass[12pt]{article}

\usepackage[T1]{fontenc}
\usepackage{bbm}

\usepackage{amsfonts}
\usepackage{graphicx}
\usepackage{verbatim}
\usepackage{enumerate}
\usepackage{amsthm}
\usepackage[english]{cleveref}
\usepackage{amsmath}

\usepackage{fullpage}
\usepackage[numbers]{natbib}

\usepackage{pxfonts}
\allowdisplaybreaks

\newtheorem{theorem}{Theorem}[section]
\numberwithin{theorem}{section}
\newtheorem{corollary}[theorem]{Corollary}
\newtheorem{lemma}[theorem]{Lemma}

\newtheorem{definition}[theorem]{Definition}

\theoremstyle{definition}

\numberwithin{equation}{section}

\newcommand{\E}{\mathbb{E}}

\newcommand{\G}{\mathbb{G}}
\newcommand{\PP}{{\bf P}}
\newcommand{\R}{\mathbb{R}}

\newcommand{\Z}{\mathbb{Z}}
\newcommand{\N}{\mathbb{N}}
\newcommand{\HH}{\mathcal{H}}
\newcommand{\EEE}{\mathcal{E}}
\newcommand{\WW}{{\bf W}}
\newcommand{\GG}{{\bf G}}
\newcommand{\TT}{{\bf T}}
\newcommand{\SSS}{{\bf S}}
\newcommand{\KK}{\mathcal{K}}
\newcommand{\I}{\mathbb{I}}
\newcommand{\Ss} {{\mathcal S}}
\newcommand{\PPP} {{\mathbb P}}
\newcommand{\RRR} {{\mathcal R}}
\newcommand{\XX}{{\bf X}}
\newcommand{\YY}{{\bf Y}}
\newcommand{\WWW}{\mathcal{W}}
\newcommand{\QQ} {{\mathcal Q}}
\newcommand{\uu}{{\bf u}}
\newcommand{\ww}{{\bf w}}

\newcommand{\UU}{{\bf U}}
\newcommand{\VV}{{\bf V}}
\newcommand{\FFF}{{\bf F}}

\newcommand{\ZZ}{{\bf Z}}
\newcommand{\HHH}{{\bf H}}

\def\FF{\mathcal F}
\def\KK{\mathcal K}

\def\LL{\mathcal L}
\def\GGG{\mathcal G}

\def\diag{\mathrm{diag}}
\def\hhh{\mathrm{h}}
\def\reg{\mathrm{reg}}
\def\cut{\mathrm{cut}}
\def\tw{\mathrm{tw}}

\def \P {{\mathcal P}}

\def \_reg {\rightarrow_{\bf reg}}

\def\maxdeg/{\Delta}

\allowdisplaybreaks

\newcommand{\du}{\mathrm{d}}

\allowdisplaybreaks

\begin{document}

\title{\textbf{On the Complexity of Nondeterministically Testable Hypergraph Parameters}}
\author{Marek Karpinski\thanks{Dept. of Computer Science and the Hausdroff Center for Mathematics, University of Bonn. Supported in part by DFG grants, the Hausdorff grant EXC59-1/2. Research partly supported by Microsoft Research New England. E-mail: \textrm{marek@cs.uni-bonn.de}}
\and 
Roland Mark\'o\thanks{Hausdorff Center for Mathematics, University of Bonn. Supported in part by a Hausdorff scholarship. E-mail: \textrm{roland.marko@hcm.uni-bonn.de}}}

\date{}
\maketitle
\begin{abstract}

{\normalsize The paper proves the equivalence of the notions of nondeterministic and deterministic parameter testing for uniform dense hypergraphs of arbitrary order. It generalizes the result previously  known only for the case of simple graphs. By a similar method we establish also the equivalence between nondeterministic and deterministic hypergraph property testing, answering the open problem in the area. We introduce a new notion of a cut norm for hypergraphs of higher order, and employ regularity techniques combined with the ultralimit method.}
\end{abstract}

\section{Introduction}

Hypergraph parameters are real-valued functions defined on the space of uniform hypergraphs of some given order invariant under relabeling the vertex set. Testing a parameter value associated to an instance in the dense model means to produce an estimation by only having access to a small portion of the data that describes it. The test data is selected by choosing a uniform random subset of the vertex set and exposing the induced substructure of the hypergraph on this subset. A certain parameter is said to be testable if for every given tolerated error the estimation is within the error range of the parameter value with high probability, and the size of the selected random subset does only depend on the size of this permitted error and not on the size of the instance, precise definitions are provided below. Similar notions apply to testing graph properties, in that situation one also uses uniform sampling in order to separate the cases where an instance has the property or is far from having it, where the distance is measured by the number of edge  modifications required. For the related notions of approximation theory and limits see \cite{AVKK2}, \cite{AKK2}, and \cite{BCL}. The general reader is referred to \cite{KM1}, \cite{L}, and \cite{LSzreg} for some related developments.

The notion nondeterministic testability was introduced by \citet{LV} in the framework of graph property testing, and encompasses an a priori weaker characteristic than the original testability. They defined that a certain property is nondeterministically testable if there exists another property of colored (edge or node) graphs that is testable in the normal sense and serves as a certificate for the original. It was shown by the authors of \cite{LV} that for graph properties the two notions are equivalent, demonstrating that if a property is nondeterministically testable, then it is also testable. Their proof used the machinery of graph limits and for this reason it was of non-effective nature. Subsequently, an explicit construction of a tester was given by \citet{GS} for nondeterministically testable graph properties containing the tester of the colored witness property as a subroutine. They used Szemer\'edi's Regularity Lemma combined with developments by \citet{AFNS}, and provided a tower-type dependence between the sample complexity of the investigated property the sample complexity of the witness property. 

In \cite{LV}, additionally the study nondeterministic testing for parameters was initiated, the definition is similar to the property testing situation. A different  approach by \citet{KM2} relying on weaker regularity methods led to an effective upper bound on the sample complexity that is a $3$-fold iteration of the exponential function applied to the sample size required by the witness parameter. 

The previous works mentioned above dealt with graphs, it was asked in \cite{LV} if the concept can be employed for hypergraphs. The notion of an $r$-uniform hypergraph (in short, $r$-graph) parameter and its testability can be defined completely analogously to the graph case, the same applies for nondeterministic testability. Naturally, first the question arises whether or not the deterministic and the nondeterministic testability are equivalent for higher order hypergraphs, and secondly, if the answer to the first question is positive, then what can be said about the relationship of the sample complexity of the parameter and that of its witness parameter. The statements that are analogous to the main results of \cite{GS}, \cite{KM2}, and \cite{LV} do not follow immediately for uniform hypergraphs of higher order from the proof for graphs, like-wise to the generalizations of the Regularity Lemma new tools and notions are required to handle these cases. In the current paper we prove the equivalence of the two testability notions for uniform hypergraphs of higher order and settle the first question posed above. Unfortunately, we were not able to obtain an explicit upper bound for the sample complexity, this is the consequence of us applying of the limit theory for hypergraphs developed by \citet{ESz} using methods of non-standard analysis, therefore the second problem still remains open. We also show that testing nondeterministically testable properties is as hard as parameter testing with our method in the sense that the same complexity bounds apply.

The paper is organized as follows. In \Cref{ch6:sec.prel} we give the preliminaries required and formulate the precise definitions followed by our main result, \Cref{ch6:main}. \Cref{ch6:sec.cut} contains the testability results for $r$-cut norms together with a brief summary of the notions and results regarding the ultralimit method that are needed for our purposes, \Cref{ch6:sec.aux} comprises some auxiliary results required. \Cref{ch6:sec.main} describes the proof of our main result. In \Cref{ch6:sec.prop} we give an application for property testing of hypergraphs, and in \Cref{ch6:sec.further} we pose some questions related to possible further research.

\section{Preliminaries and main result}\label{ch6:sec.prel}


A simple $r$-uniform hypergraph on $n$ vertices is a subset $G$ of ${[n] \choose r}$, the size of $G$ is $n$, and the elements of ${[n] \choose r}$ are $r$-edges. Let $A_G$ denote the symmetric $\{0,1\}$-valued $r$-array or symmetric subset of $[n]^r$ that represents $G$, we will sometimes use also only the term $G$ to refer to a symmetric subset of $[n]^r \setminus \diag([n]^r)$ corresponding to the array representation. Let $k$ be a positive integer, and let $\GGG_n^{r,k}$ denote the set of {\it $k$-colored} $r$-uniform hypergraphs of size $n$, that are partitions $\GG=(G^\alpha)_{\alpha \in [k]}$ of ${[n] \choose r}$ into $k$ classes, so in all what follows here colored $r$-graph means a complete $r$-graph where to each edge $e$ we assign exactly one color $\GG(e)$ from the set $[k]$. In this sense simple $r$-graphs are regarded as $2$-colored. In the $k$-colored case is also possible to speak about the array representation, $A_{G^\alpha}$ stands for the symmetric $\{0,1\}$-valued $r$-array that represents the color class of $\alpha$, again with slight abuse of notation we will use $G^\alpha$ for $A_{G^\alpha}$. Additionally we have to introduce the color $\iota$ and the corresponding array $A_\iota$ that always is the indicator array of the set of diagonal elements of $[n]^r$ (those having repetitions in their coordinates, denoted by $\diag([n]^r)$). For any finite set $C$ the term $C$-colored graph is defined analogously.

A {\it $k$-coloring} of a  $t$-colored $r$-graph $\GG=(G^\alpha)_{\alpha \in [t]}$ is a $tk$-colored $r$-graph $\hat \GG=(G^{(\alpha,\beta)})_{\alpha \in [t], \beta\in [k] }$ with colors from the set $[t] \times [k]$, where each of the original color classes indexed by $\alpha \in [t]$ is retrieved by taking the union of the new classes corresponding to $(\alpha,\beta)$ over all $\beta \in [k]$, that is $G^\alpha=\cup_{\beta \in [k]}G^{(\alpha,\beta)}$. This last operation is called {\it $k$-discoloring} of a $[t] \times [k]$-colored graph, we denote it by $[\hat \GG,k]=\GG$.
We will sometimes write $tk$-colored for $[t] \times [k]$-colored graphs when it is clear from the context what we mean. 

Further, for  a finite set $S$, let $\hhh(S)$ denote the set of nonempty subsets of $S$, 
 and $\hhh(S,m)$ the set of nonempty subsets of $S$ of cardinality at most $m$. A real $2^r-1$-dimensional vector $x_{\hhh(S)}$ denotes $(x_{T_1}, \dots, x_{{T_{2^r-1}}})$, where $T_1, \dots, T_{2^r-1}$ is a fixed ordering of the nonempty subsets of $S$ with $T_{2^r-1}=S$, for a permutation $\pi$ of the elements of $S$ the vector  $x_{\pi(\hhh(S))}$ means $(x_{\pi'(T_1)}, \dots, x_{\pi'({T_{2^r-1}})})$, where $\pi'$ is the action on the subsets of $S$ induced by $\pi$. 

We will require some basic notation from graph limit theory, and we summarize their relevance outlined in previous works, \citet{L} is a comprehensive reference for the area.

Let $q \geq 1$ and $\GG \in \GGG_n^{r,k}$, then  $\G(q,\GG)$ denotes
the  random $r$-graph on $q$ vertices that is obtained by uniformly picking a random subset $S$ of $[n]$ of cardinality $q$ and taking the induced subgraph $\GG[S]$. For any $\FFF \in \GGG_q^{r,k}$ and $\GG \in \GGG^{r,k}$ the $\FFF$-density of $\GG$ is defined as $t(\FFF,\GG)=\PPP(\FFF = \G(q, \GG))$.

Let the $r$-kernel space $\mathcal{W}_0^r$ denote the space of the bounded measurable functions $W\colon [0,1]^{\hhh([r],r-1)} \to \R$, and the subspace $\mathcal{W}^r$ of $\mathcal{W}_0^r$ symmetric $r$-kernels that are invariant under coordinate permutations induced by $\pi \in S_r$, that is $W(x_{\hhh([r],r-1)})=W(x_{\pi(\hhh([r],r-1))})$ for each $\pi \in S_r$. We will refer to this invariance in the paper both for $r$-kernels and for measurable subsets of $[0,1]^{\hhh([r])}$ as {\it satisfying the usual symmetries}.
 Assume that the functions $W\in \mathcal{W}^r_{I}$ take their values in the interval $I$, for $I=[0,1]$ we call these special symmetric $r$-kernels  {\it{$r$-graphons}.} In what follows, $\lambda$ always denotes the usual Lebesgue measure in $\R^d$, where $d$ is everywhere clear from the context.

Analogously to the graph case we define the space of {\it $k$-colored $r$-graphons} $\mathcal W^{r,k}$ whose elements are referred to as $\WW=(W^\alpha)_{\alpha \in [k]}$ with each of the $W^\alpha$'s being an $r$-graphon. The special color $\iota$ that stands for the absence of any colors in the diagonal in some sense can be also employed in this setting, see below for the case when we represent a $k$-colored $r$-graph as a graphon. The corresponding $r$-graphon $W^\iota$ is $\{0,1\}$-valued. Furthermore $\sum_{\alpha \in [k]}W^\alpha(x)=1-W^{\iota}(x)$ everywhere on $[0,1]^{\hhh([r],r-1)}$. For  $x \in [0,1]^{\hhh([r])}$ the expression $\WW(x)$ denotes the color at $x$, we have $\WW(x)=\alpha$ whenever  $\sum_{i=1}^{\alpha-1} W^i(x_{\hhh([r],r-1)}) \leq x_{[r]} \leq \sum_{i=1}^{\alpha} W^i(x_{\hhh([r],r-1)})$.
  
Similar to the  finitary case, a {\it $k$-coloring} of a $\WW \in \mathcal W^{r,k}$ is a $tk$-colored $r$-graphon $\hat \WW=(W^{(\alpha,\beta)})_{\alpha \in [t], \beta\in [k] }$ with colors from the set $[t] \times [k]$ so that $\sum_{\alpha \in [t], \beta \in [k]} W^{(\alpha, \beta)}(x)=W^\alpha(x)$ for each $x \in [0,1]^{\hhh([r],r-1)}$ and $\alpha \in [t]$. The $k$-discoloring $[\hat \WW,k]$ of $\hat \WW$ and the term $C$-colored graphon is defined analogously, and simple $r$-graphons are treated as $2$-colored. 

For $q \geq 1$ and $\WW \in \mathcal W^{r,k}$ the random $[k] \cup \{\iota\}$-colored $r$-graph $\G(q, \WW)$ is generated as follows. The vertex set of $\G(q, \WW)$ is $[q]$, we have to pick uniformly a random point $(X_S)_{S \in \hhh([q], r-1)} \in [0,1]^{\hhh([q], r-1)}$, then conditioned on this choice we conduct independent trials to determine the color of each edge $e \in {[q] \choose r}$ with the distribution given by $\PPP_e(\G(q, \WW)(e) = \alpha) = W^\alpha(X_{\hhh(e,r-1)})$ corresponding to $e$. Recall that $\iota$ is a special color which we want to avoid in most cases, therefore we will highlight the conditions imposed on the above random variables so that $\G(q, \WW) \in \GGG^{r,k}$.

 For $\FFF \in \GGG_q^{r,k}$ the $\FFF$-density of $\WW$ is defined as $t(\FFF,\WW)=\PPP(\FFF = \G(q, \WW))$, which can be written following the above description of the random graph as 
\begin{align*}
t(\FFF,\WW)= \int_{ [0,1]^{\hhh([q], r-1)} } \prod_{e \in {[q] \choose r}} W^{\FFF(e)}(x_{\hhh(e,r-1)}) \du \lambda(x).
\end{align*}

The above notions were introduced in order to provide a concise representation for the limit space of $r$-graphs in \cite{ESz} and \cite{LSzlim}, in the current work we will not draw on this development explicitly but mention their relevance here. In a nutshell, a sequence of $r$-graphs converges if the corresponding numerical $\FFF$-density sequences converge for all $r$-graphs $\FFF$. One of the main results of \cite{LSzlim} for graphs and \cite{ESz} in the general case is that for every convergent sequence of $r$-graphs there exists an $r$-graphon they converge to in the sense that the $\FFF$-densities approach the $\FFF$-density of the limiting $r$-graphon. This was later reproved by \cite{Z} for general $r$ with purely combinatorial methods that are similar to concepts employed in the current paper.

We can associate to each $\GG \in \GGG_n^{r, k}$ an element $\WW_\GG \in \mathcal W^{r,k}$ by subdividing the unit cube $[0,1]^{\hhh([r],1)}$ into $n^r$ small cubes the natural way and defining the function $W' : [0,1]^{\hhh([r],1)} \to [k]$ that takes the value $\GG(\{i_1, \dots, i_r\})$ on $[\frac{i_1-1}{n}, \frac{i_1}{n}] \times \dots \times [\frac{i_r-1}{n}, \frac{i_r}{n}]$ for distinct $i_1, \dots, i_r$, and the value $\iota$ on the remaining diagonal cubes. Then set $(\WW_\GG)^\alpha(x_{\hhh([r],r-1)}) = \I(W'(p_{\hhh([r],1)}(x_{\hhh([r],r-1)}))=\alpha)$ for each $\alpha \in [k] \cup\{\iota\}$, where $p_{\hhh([r],1)}$ is the projection to the suitable coordinates. Note that 
\begin{align}\label{ch6:eq22}
|t(\FFF,\GG)-t(\FFF,\WW_\GG)|\leq \frac{{q \choose 2}}{n - {q \choose 2}}
\end{align}
 for each $\FFF \in \GGG_q^{r,k}$, hence the previous 
representation is compatible in the sense that $\lim_{n \to \infty} t(\FFF,\GG_n)=\lim_{n \to \infty} t(\FFF,\WW_{\GG_n})$ for any sequence $\{\GG_n\}_{n=1}^\infty$ with $|V(\GG_n)|$ tending to infinity.

We proceed by providing the necessary formal definitions of the parameter testability in the dense hypergraph model. 

\begin{definition}\label{ch6:def.partest}
An $r$-graph parameter $f$ is testable if for any $\varepsilon>0$ there exists a positive integer $q_f(\varepsilon)$ such that for any simple $r$-graph $G$ with at least $q_f(\varepsilon)$ nodes we have that
\begin{align*}
\PP(|f(G)-f(\G(q_f(\varepsilon),G)| >\varepsilon) < \varepsilon.
\end{align*}
The smallest function $q_f$ satisfying the previous inequality is called the sample complexity of $f$.
The testability of parameters of $k$-colored $r$-graphs is defined analogously.
\end{definition}

An a priori weaker characteristic than the one above, nondeterministic testability, is the second cornerstone of the current work, and was introduced in \cite{LV}.

\begin{definition}\label{ch6:def.ndpartest}
An $r$-graph parameter $f$ is non-deterministically testable if there exist an integer $k$ and a testable $2k$-colored directed $r$-graph parameter $g$ called witness such that for any simple graph $G$ the value $f(G)=\max_{\GG} g(\GG)$ where the maximum goes over the set of $k$-colorings of $G$ (regarded as an element of $\GGG^{r,2}$). 
\end{definition}

Originally in \cite{LV}, the witness parameter was a function of $k$-colored graphs, and the maximum was taken over the set of $(k,m)$-colorings of the original graph in order to determine the parameter value, meaning that the present edges are colored by elements of $[m]$, absent ones by the remaining colors from $[k] \setminus [m]$. Our modification is equivalent to that setting and is motivated by notational purposes.


In the current paper we only deal with undirected structures, but similar results can be obtained when the witness parameter is defined on the space of directed $r$-graphs. In this case, in order to obtain $G$ from $\GG$ as above after the discoloring we additionally have to undirect the edges and neglect multiplicities created by the former operation.

The maximization expression  in \Cref{ch6:def.ndpartest} is somewhat arbitrary and could be replaced for example by minimization, this would however not affect the testability characteristic of the parameter. Our main result extends the equivalence of the two testability notions for arbitrary $r$, this was first proved by \citet{LV} for $r=2$.

\begin{theorem}\label{ch6:main}
Every non-deterministically testable $r$-graph parameter $f$ is testable.
\end{theorem}

Our proof follows the proof skeleton of \cite{KM2}, but requires a more sophisticated approach.The reason for this is that the  analogous norm  for hypergraphs to the cut-norm that comes with a counting lemma has some shortcomings, for instance the sample is in some cases far away from the original in the natural distance notion induced by the norm. Therefore the corresponding regularity lemma cannot be applied directly as in \cite{KM2}.

The definitions of the relevant norms is given next.

\begin{definition} \label{ch6:def.cutnorm}

Let $r \geq 1$ and $A$ be a real $r$-array of size $n$. Then the cut norm of $A$ is
\begin{align*}
\|A\|_{\square,r}=\frac{1}{n^r} \max_{\substack{S_i \subset [n]^{r-1} \setminus \diag ([n]^{r-1}) \\ i \in [r]}}
|A(r;S_1, \dots, S_r)|,
\end{align*}
where $A(r;S_1, \dots, S_r)=\sum_{i_1, \dots, i_r =1}^n A(i_1, \dots, i_r) \prod_{j=1}^r \I_{S_j} (i_1, \dots, i_{j-1}, i_{j+1}, \dots, i_r)$, and the maximum goes over sets $S_i$ that are invariant under coordinate permutations. 

If $\P=(P_i)_{i=1}^t$ is a partition of $[n]^{r-1}\setminus \diag ([n]^{r-1})$ with symmetric classes, then the cut-$\P$-norm of $A$ is
\begin{align*}
\|A\|_{\square,r,\P}=\frac{1}{n^r} \max_{\substack{S_i \subset [n]^{r-1}}, i \in [r]}
\sum_{j_1, \dots, j_r=1}^t |A(r;S_1\cap P_{j_1} , \dots, S_r \cap P_{j_r})|.
\end{align*}

The cut norm of an $r$-kernel $W$ is
\begin{align*}
\|W\|_{\square,r}=\sup_{\substack{S_{i} \subset [0,1]^{\hhh([r-1])} \\ i \in [r]}}
|\int_{ \cap_{i \in [r]} p_{[r]\setminus \{i\}}^{-1}(S_i)}W(x_{\hhh([r],r-1)}) \du \lambda(x_{\hhh([r],r-1)}) |,
\end{align*}
where the supremum is taken over sets $S_i$ that satisfy the usual symmetries, and $p_e$ is the natural projection from $[0,1]^{\hhh([r],r-1)}$ onto $[0,1]^{\hhh(e)}$. 
Furthermore, for a symmetric partition $\P=(P_i)_{i=1}^t$ of $[0,1]^{\hhh([r-1])}$ the cut-$\P$-norm of an $r$-kernel is defined by
\begin{align*}
\|W\|_{\square,r,\P}=\sup_{\substack{S_{i} \subset [0,1]^{\hhh([r-1])} \\ i \in [r]}}
\sum_{j_1, \dots, j_r=1}^t |\int_{\cap_{i \in [r]} p_{[r]\setminus \{i\}}^{-1}(S_i \cap P_{j_i} ) } W(x_{\hhh([r],r-1)}) \du \lambda(x_{\hhh([r],r-1)}) |,
\end{align*}
where the supremum is taken over sets $S_i$ that satisfy the usual symmetries.
\end{definition}
We remark that it is also true that
\begin{align*}
\|W\|_{\square,r}=\sup_{f_1, \dots, f_r \in \mathcal [0,1]^{\hhh([r-1])}}
|\int\limits_{[0,1]^{\hhh([r],r-1)}} \prod_{i=1}^r f_i(x_{\hhh([r]\setminus \{i\})}) W(x_{\hhh([r],r-1)}) \du \lambda(x_{\hhh([r],r-1)}) |,
\end{align*}
where the supremum goes over functions $f_i$ that satisfy the usual symmetries, and similarly for any  symmetric partition $\P=(P_i)_{i=1}^t$ of $[0,1]^{\hhh([r-1])}$ we have with the same conditions for the $f_i$'s as above that
\begin{align*}
\|W\|_{\square,r}=\sup_{f_1, \dots, f_r \in \mathcal [0,1]^{\hhh([r-1])}}
\sum_{j_1, \dots, j_r=1}^t  |\int\limits_{[0,1]^{\hhh([r],r-1)}} \prod_{i=1}^r f_i(x_{\hhh([r]\setminus \{i\})})\I_{P_{j_i}} (x_{\hhh([r]\setminus \{i\})}) W(x_{\hhh([r],r-1)}) \du \lambda(x_{\hhh([r],r-1)}) |.
\end{align*}

In several previous works, see e.g. \cite{AVKK2}, the cut norm for $r$-arrays denotes a term that is significantly different from the one in \Cref{ch6:def.cutnorm} and is not suitable for our present purposes.
The above norms give rise to a distance between $r$-graphons, and analogously for $r$-graphs.

\begin{definition}
For two $k$-colored $r$-graphons $\UU=(U^\alpha)_{\alpha \in [k]}$ and $\WW=(W^\alpha)_{\alpha \in [k]}$ their cut distance is defined as 
\begin{equation*}
d_{\square,r}(\UU,\WW) = \sum_{\alpha=1}^k \|U^\alpha-W^\alpha\|_{\square,r},
\end{equation*}
and their cut-$\P$-distance as
\begin{equation*}
d_{\square,r,\P}(\UU,\WW) = \sum_{\alpha=1}^k \|U^\alpha-W^\alpha\|_{\square,r,\P}.
\end{equation*}
For two $k$-colored $r$-graphs $\GG=(G^\alpha)_{\alpha \in [k]}$ and $\HHH=(H^\alpha)_{\alpha \in [k]}$ their corresponding distances are defined as 
\begin{align*}
d_{\square,r}(\GG,\HHH)=d_{\square,r}(\WW_\GG,\WW_\HHH),
\end{align*}
and
\begin{align*}
d_{\square,r,\P}(\GG,\HHH)=d_{\square,r,\P}(\WW_\GG,\WW_\HHH).
\end{align*}
Distances between an $r$-graph and an $r$-graphon, as well as for $r$-kernels, is analogously defined.
\end{definition}

Note that the norms introduced above are in general smaller or equal than the $1$-norm of integrable functions, also $d_{\square,r}(\UU,\WW)  \leq
d_{\square,r,\P}(\UU,\WW)$ hods for every pair. 
Their relevance will be clearer in the context of the next counting lemma, we include the standard proof only for completeness' sake.

\begin{lemma}
\label{ch6:hypercount}
Let $\UU$ and $\WW$ be two $k$-colored $r$-graphons with $\|U\|_\infty, \|W\|_\infty \leq 1$. Then for every $\FFF \in \GGG_q^{r,k}$ it holds that
\begin{align*}
|t(\FFF,\WW)-t(\FFF,\UU)| \leq {q \choose r} d_{\square,r}(\UU,\WW). 
\end{align*}
\end{lemma}
\begin{proof}
Fix $q$ and $\FFF \in \GGG_q^{r,k}$. Then
\begin{align*}
&|t(\FFF,\WW)-t(\FFF,\UU)| = |\int\limits_{ [0,1]^{\hhh([q], r-1)} } \prod_{e \in {[q] \choose r}} W^{\FFF(e)}(x_{\hhh(e,r-1)}) -\prod_{e \in {[q] \choose r}} U^{\FFF(e)}(x_{\hhh(e,r-1)}) \du \lambda(x)| \\
& \quad \leq \sum_{ e \in {[q] \choose r}} | \int\limits_{ [0,1]^{\hhh([q], r-1)} }  [W^{\FFF(e)}(x_{\hhh(e,r-1)})-U^{\FFF(e)}(x_{\hhh(e,r-1)}) ] \\ & \qquad \qquad \qquad \qquad \prod_{f \in {[q] \choose r}, f \prec e} W^{\FFF(f)}(x_{\hhh(f,r-1)})\prod_{g \in {[q] \choose r}, e \prec g} U^{\FFF(g)}(x_{\hhh(g,r-1)})\du \lambda(x)| \\
& \quad \leq  \sum_{ e \in {[q] \choose r}} \|W^{\FFF(e)}-U^{\FFF(e)}\|_{\square,r} \leq {q \choose r} d_{\square,r}(\UU,\WW),
\end{align*}
where $\prec$ is an arbitrary total ordering of the elements of ${q \choose r}$.
\end{proof}

Let $d_\tw$ denote the total variation distance between probability measures on $\GGG_n^{r,[k]*}$, where $[k]*=[k] \cup \{\iota\}$ for $k \geq 1$ (without highlighting the specific parameters in the notion $d_\tw$), that is $d_\tw(\mu, \nu)=\max_{\FF \subset \GGG_n^{r,[k]*}} |\mu(\FF)-\nu(\FF)|$,  and let the measure $\mu(q,\GG)$, respectively $\mu(q, \WW)$, denote the probability measure of the random $r$-graph $\G(q, \GG)$, respectively $\G(q,\WW)$, taking values in $\GGG_q^{r,[k]*}$. 
It is a standard observation then that 
\begin{equation}\label{ch6:eqtw}
d_\tw (\mu(q, \WW), \mu(q, \UU)) = \frac{1}{2} \sum_{\FFF \in \GGG_q^{r,[k]*}} |t(\FFF, \WW)-t(\FFF, \UU)|,
\end{equation}
and that $\G(q, \WW)$ and $\G(q, \UU)$ can be coupled in form of the random $r$-graphs $\GG_1$ and $\GG_2$, such that 
\begin{align}\label{ch6:eq101}
d_\tw (\mu(q, \WW), \mu(q, \UU)= \frac{1}{2} \PP(\GG_1 \neq \GG_2),
\end{align}
 and further, for any coupling $\GG'_1$ and $\GG'_2$  it hods that $d_\tw (\mu(q, \WW), \mu(q, \UU) \leq  \frac{1}{2} P(\GG'_1 \neq \GG'_2)$.

For $\GG \in \GGG_n^{r,k}$ note that 
\begin{align}\label{ch6:eq100}
d_\tw (\mu(q, \GG), \mu(q, \WW_\GG)) \leq q^2/n,
\end{align}
where the right hand side is a simple upper bound on the probability that if we uniformly choose $q$ elements of an $n$-element set, then we get at least two identical objects. The inequality (\ref{ch6:eq100}) follows from the fact that conditioned on the event that the independent and uniform $X_{\{i\}}$'s for $i \in [q]$ fall in different intervals $[\frac{j-1}{n},\frac{j}{n}]$ for $j \in [n]$ the distribution of $\G(q, \WW_\GG)$ is the same as the distribution of $\G(q, \GG)$.

The next corollary is a direct consequence of \Cref{ch6:hypercount}.

\begin{corollary} \label{ch6:hypercountcor}
If $\UU$ and $\WW$ are two $k$-colored $r$-graphons, then 
\begin{align*}
d_\tw (\mu(q, \WW), \mu(q, \UU) \leq \frac{k^{q^r} q^r}{2r!} d_{\square,r}(\UU,\WW),
\end{align*}
and there exists a coupling in form of $\GG_1$ and $\GG_2$ of the random $r$-graphs $\G(q, \WW)$ and $\G(q, \UU)$, such that 
\begin{align*}
P(\GG_1 \neq \GG_2) \leq \frac{k^{q^r} q^r}{2r!} d_{\square,r}(\UU,\WW). 
\end{align*}
\end{corollary}

A generalization of the notion of a step function in the case of graphs to the situation where we deal with $r$-graphs is given next. For a partition $\P$ the number of its classes is denoted by  $t_\P$.

\begin{definition}
We call an $k$-colored $r$-graphon $\WW$ with $r \geq l$ an $(r,l)$-step function if there exist  positive integers $t_l, t_{l+1}, \dots, t_r=k$, symmetric partitions $\P=(P_1, \dots, P_{t_l})$ of $[0,1]^{\hhh([l])}$, and real arrays $A^\alpha_s: [t_{s-1}]^{\hhh([s],s-1)} \to [0,1]$  with $\alpha \in [t_s]$ for $l \leq s \leq r$ such that $\sum_{\alpha \in [t_s]} A^\alpha_s(i_{\hhh([s],s-1)})=1$ for any choice of $i_{\hhh([s],s-1)}$ and for $s \leq r$ so that $W^\alpha$ for $\alpha \in [k]$ is of the following form for each $\alpha \in [k]$.
\begin{align*}
& W^\alpha(x_{\hhh([r])}) =  \sum_{\substack{i_S=1 \\ S \subset [r], l \leq |S|}}^{t_{|S|}} A^\alpha_r(i_{\hhh([r],r-1)}) \\ & \qquad\prod_{S \in {[r] \choose l}} \I_{P_{i_S}} (x_{\hhh(S)}) \prod_{\substack{S \subset [r] \\ l+1 \leq |S| < r}} \I( \sum_{j=1}^{i_{S}-1}A^j_{|S|}(i_{\hhh(S, |S|-1)})\leq x_S \leq \sum_{j=1}^{i_{S}}A^j_{|S|}(i_{\hhh(S, |S|-1)})).
\end{align*}
We refer to the partition $\P$ as the steps of $W$.
\end{definition}
The most simple example is the $(r,r-1)$ step function that can be written as 
\begin{equation*}
W^\alpha(x_{\hhh([r])}) =  \sum_{i_1, \dots , i_r=1}^{t_{r-1}} A^\alpha_r(i_1,\dots, i_r)\prod_{j =1}^r \I_{P_{i_j}} (x_{\hhh([r] \setminus \{j\})}).
\end{equation*}

\section{Testability of the $r$-cut norm}\label{ch6:sec.cut}
We define a parameter of $r$-uniform hypergraphs that is a generalization of the ground state energies of \cite{BCL2} in the case of graphs. This notion encompasses several important quantities, therefore its testability is central to many applications.  

\begin{definition}\label{ch6:def.gse1}
For a set $H \subset {[n] \choose r}$, a real $r$-array $J$ of size $q$, and a symmetric partition $\P=(P^1, \dots, P^q)$ of ${[n] \choose r-1}$ we define the energy
\begin{equation*}
\EEE_{\P,r-1} (H,J)=\frac{1}{n^r} \sum_{i_1, \dots, i_r =1}^q J(i_1, \dots, i_r) e_{H} (r; P_{i_1}, \dots, P_{i_r}),
\end{equation*}
where $e_{H}(r; S_1, \dots, S_r)= | \{(u_1, \dots, u_r) \in [n]^r | A_H(u_1, \dots, u_r)=1 \textrm{ and } A_{S_j}(u_1, \dots, u_{j-1}, u_{j+1}, \dots, u_r)=1  \textrm{ for all } j=1, \dots, r\} |$. 

Let $\HHH=(H^\alpha)_{\alpha \in [k]}$ be a $k$-colored $r$-uniform hypergraph on the vertex set $[n]$ and $J^{\alpha}$ a be real $q \times \dots \times q$ $r$-array with $\| J \|_\infty \leq 1$ for each $\alpha \in [k]$. Then the energy for a partition $\P$ as above is 
\begin{equation*}
\EEE_{\P,r-1} (\HHH,J)= \sum_{\alpha \in [k]} \EEE_{\P,r-1} (H^\alpha,J^\alpha).
\end{equation*}

The maximum of the energy over all partitions $\P$ of ${[n] \choose r-1}$ is called the ground state energy (GSE) of $H$ with respect to $J$, and is denoted by 
\begin{equation*}
\EEE_{r-1} (\HHH,J)=\max_\P \EEE_{\P,r-1} (\HHH,J).
\end{equation*}
\end{definition}

The GSE can also be defined for $r$-graphons.

\begin{definition}\label{ch6:def.gse2}
For an $r$-graphon $W$, a real $r$-array $J$ of size $q$, and a symmetric partition $\P=(P^1, \dots, P^q)$ of $[0,1]^{\hhh([r-1])}$ we define the energy
\begin{equation*}
\EEE_{\P,r-1} (W,J)= \sum_{i_1, \dots, i_r \in [q]}J(i_1, \dots, i_r) \int_{ \cap_{j \in [r]} p_{[r]\setminus \{j\}}^{-1}(S_{i_j})} W(x_{\hhh([r],r-1)}) \du \lambda(x_{\hhh([r],r-1)}).
\end{equation*}

Let $\WW=(W^\alpha)_{\alpha \in [k]}$ be a $k$-colored $r$-graphon and $J^{\alpha}$ a be real $q \times \dots \times q$ $r$-array with $\| J \|_\infty \leq 1$ for each $\alpha \in [k]$. Then the energy for a partition $\P$ as above is 
\begin{equation*}
\EEE_{\P,r-1} (\WW,J)= \sum_{\alpha \in [k]} \EEE_{\P,r-1} (W^\alpha,J^\alpha).
\end{equation*}
and the GSE of $\WW$ with respect to $J$, and is denoted by 
\begin{equation*}
\EEE_{r-1} (\WW,J)=\sup_\P \EEE_{\P,r-1} (\WW,J),
\end{equation*}
where the supremum runs over all symmetric partitions $\P=(P^1, \dots, P^q)$ of $[0,1]^{\hhh([r-1])}$.
\end{definition}

Definitions of the above energies are analogous in the directed, and the weighted case, and also for $r$-kernels. The next lemma tells us about the distribution of the GSE when taking a random sample $\G(n,\HHH)$ of an $\HHH \in \GGG^{r,k}$. 

\begin{lemma} \label{ch5:conc} 
The expression $\EEE_{r-1} (\G(n, \HHH),J)$ is highly concentrated around its mean, that is for every $\varepsilon >0$ it holds that 
$$
\PP(|\EEE_{r-1} (\G(n, \HHH),J)-\E \EEE_{r-1} (\G(n, \HHH),J)| \geq \varepsilon  \|J\|_\infty ) \leq 2 \exp(-\frac{\varepsilon^2 n}{8 r^2}).
$$
\end{lemma}
\begin{proof}
We can assume that $\|J\|_\infty\leq 1$. The random $r$-graph $\G(n, \HHH)$ is generated by picking random nodes from $V(\HHH)$ without repetition, let $X_i$ denote the $i$th random element of $V(\HHH)$ that has been selected. Define the martingale $Y_i=\E [\EEE_{r-1} (\G(n, \HHH),J) | X_1, \dots, X_i]$ for $0 \leq i \leq n.$ It has the property  that $Y_0=\E [\EEE_{r-1} (\G(n, \HHH),J)]$ and $Y_n=\EEE_{r-1} (\G(n, \HHH),J)$, whereas the jumps $|Y_i-Y_{i-1}|$ are bounded above by $\frac{2r}{n}$ for each $i \in [n]$. The last observation is the consequence of the fact that for any partition $\P$ of ${[n] \choose r-1}$ only at most $r n^{r-1}$ terms in the sum constituting $\EEE_{\P,r-1} (\HHH,J)$ are affected by changing the placing of $X_{i+1}$ in the classes of $\P$. Applying the Azuma-Hoeffding inequality to the martingale verifies the statement of the lemma. 
\end{proof}

The same concentration result as above applies also to $\EEE_{r-1} (\G(n, \WW),J)$. 

We will show that these hypergraph parameters are testable via the ultralimit method and the machinery developed by \citet{ESz}. From the notational perspective and theoretical background this section slightly stands out from the rest of the paper. First we give a brief summary of the notions that were used in \cite{ESz} in order to produce a representation for the limit space of simple $r$-graphs. This representation led to a new analytical proof method for several results for simple $r$-graphs such as the Regularity Lemma, the Removal Lemma, or the testability assertion about hereditary $r$-graph properties. Subsequently, technical results proved in \cite{ESz} which are relevant here are mentioned, for more details and complete proofs we refer to the source paper \cite{ESz}. 

Recall that a sequence of $r$-graphs $(G_n)_{n\geq 1}$ is convergent if for every simple $F$ the numerical sequences $t(F,G_n)$ converge when $n$ tends to infinity.

We start by introducing the basic notations for ultraproduct measure spaces. Let us fix a non-principal ultrafilter $\omega$ on $\N$, and let $X_1, X_2, \dots$ be a sequence of finite sets of increasing size. We define the infinite product set $\hat X = \prod_{i=1}^\infty X_i$ and the equivalence relation $\sim$ between elements of $\hat X$, so that $p \sim q$ if and only if $\{i | \quad p_i=q_i\} \in \omega$. Set $\XX=\hat X / \sim$, this set is called the ultraproduct of the $X_i$'s, and it will serve as the base set of the ultraproduct probability space. Further, let $\P$ denote  the algebra of subsets of $\XX$ of the form $A=[\{A_i\}_{i=1}^\infty]$, where $A_i \subset X_i$ for each $i$, and $[.]$ denotes the equivalence class under $\sim$ (for convenience, $p=[\{p_i\}_{i=1}^\infty] \in [\{A_i\}_{i=1}^\infty]$ exactly in the case when $\{i | p_i \in A_i \} \in \omega$). 

Define a measure on the sets  belonging to $\P$ through the ultralimit of the counting measure on the sets $X_i$, that is, $\mu(A)=\lim_\omega \frac{|A_i|}{|X_i|}$, where the ultralimit of a bounded real numerical sequence $\{x_i\}_{i=1}^\infty$ is denoted by $x=\lim_\omega x_i$, and is defined by the property that for every $\varepsilon >0$ we have $\{i | \quad  |x-x_i| < \varepsilon \} \in \omega$. One can see that the limit exists for every bounded sequence and is unique, therefore well-defined, this is a consequence of basic properties of a non-principal ultrafilter. The set of $\mathcal N \subset 2^\XX$ of $\mu$-null sets is the family of sets $N$ for those there exists an infinite sequence of supersets $\{A^i\}_{i=1}^\infty \subset \P$ such that $\mu(A^i) \leq 1/i$. Finally define the $\sigma$-algebra $\mathcal B$ on $\XX$ by the $\sigma$-algebra generated by $\P$ and $\mathcal N$, and set the measure $\mu(B)=\mu(A)$ for each $B \in \mathcal B$,  where $A \triangle B \in \mathcal N$ and $A \in\P$. Again, everything is well-defined, see \cite{ESz}, so we arrive at the ultraproduct measure space $(\XX, \mathcal B, \mu)$. 

Let $X_1, X_2, \dots$ and $Y_1, Y_2, \dots$ be two increasing sequences of finite sets with ultraproducts $\XX$ and $\YY$ respectively, then it is true that the ultraproduct of the product sequence $X_1 \times Y_1, X_2 \times Y_2, \dots$ is the product $\XX \times \YY$, but the $\sigma$-algebra $\mathcal B_{\XX \times \YY}$ of the measure space can be strictly larger than the $\sigma$-algebra generated by $\mathcal B_\XX \times \mathcal B_\YY$, and this is a crucial point when the aim is to construct a separable representation of the ultraproduct measure space of product sets.

Let $r$ be some positive integer, and again $X_1, X_2, \dots$ a sequence of finite sets as above. For any $e \subset [r]$ we define the ultraproduct measure spaces $(\XX^e, \mathcal B_{\XX^e}, \mu_e )$, also let $P_e$ denote the natural projection from $\XX^{[r]}$ to $\XX^e$. Furthermore let $\sigma(e)$ denote the sub-$\sigma$-algebra of $\mathcal B_{\XX^{[r]}}$ given by $P_e^{-1}(\mathcal B_{\XX^e})$, and $\sigma(e)^*$ be the sub-$\sigma$-algebra $\langle P_f^{-1}(\mathcal B_{\XX^f}) | f \subset e, |f| < |e|  \rangle$. Note that in general  $\sigma(e)$ is strictly larger than $\sigma(e)^*$. We denote the measure $\mu_{\XX^{e}}$ simply by $\mu_{e}$ and the $\sigma$-algebra $\mathcal B_{\XX^e}$ by $\mathcal B_{e}$.

\begin{definition}
Let $r$ be a positive integer. We call a measure preserving map $\phi: \XX^{[r]} \to [0,1]^{\hhh([r])}$ a separable realization if 
\begin{enumerate}
\item for any permutation $\pi \in S_{[r]}$ of the coordinates  we have for all $x \in \XX^{[r]}$ that $\Pi(\phi(x))=\phi(\phi(x))$, where $\Pi$ is the permutation of the power set of $[r]$ induced by $\pi$, and
\item for any $e \subset \hhh([r])$ and any measurable $A \subset [0,1]$ we have that $\phi^{-1}_e(A) \in \sigma(e)$ and $\phi^{-1}_e(A)$ is independent of $\sigma(e)^*$.
\end{enumerate}
\end{definition}

We are interested in the limiting behavior of sequences of $k$-partitions (or edge-$k$-colored $r$-graphs on the vertex sets $X_1, X_2, \dots$) of the sequence $X^r_1, X^r_2, \dots$, where convergence is defined in the following general way.

 Let $G_i=(G^1_i, \dots, G^k_i)$ be a symmetric partition of $X_i^r$ for each $i \in \N$, then $(G_i)_{i=1}^\infty$ converges if for every $k$-colored $r$-graph $F$ the numerical sequences $t(F,G_i)$ converge, as in \Cref{ch6:sec.prel}. The ultralimit method enables us to handle the cases where the convergence does not hold without going to subsequences, we describe the method next. Let us denote the size of $F$ by $m$ and let $F(e)$ be the color of $e \in {[m] \choose r}$, then $t(F,G_i)$ can be written as the measure of a subset of $X_i^m$. We show this by explicitly presenting the set denoted by $T(F,G_i)$, so let
\begin{align}\label{ch6:eq1}
T(F,G_i)= \bigcap_{e \in {[m] \choose r}} P_e^{-1} (P_{s_e} (G_i^{F(e)})),
\end{align} 
where $P_e$ is the natural projection from $X_i^{[m]}$ to $X_i^e$, and $P_{s_e}$ is a bijection going from  $X_i^{[r]}$ to $X_i^{e}$ induced by an arbitrary but fixed bijection $s_e$ between $e$ and $[r]$.
We define the induced  subgraph density of the ultraproduct of $k$-colored $r$-graphs formally following \eqref{ch6:eq1}, if  $\GG=(\GG^1, \dots, \GG^k)$ is a $\mathcal B_{[r]}$-measurable $k$-partition of $\XX^{[r]}$ and $F$ is as above then let 
\begin{align}\label{ch6:eq2}
T(F,\GG)= \bigcap_{e \in {[m] \choose r}} P_e^{-1} (P_{s_e} (\GG^{F(e)})).
\end{align}

It is easy to see that $\lambda(T(F,G_i))=t(F,G_i)$. Forming the ultraproduct of a series of sets commutes with finite intersection, therefore $\lim_\omega T(F,G_i) = T(F,\lim_\omega G_i)$ and $\lim_\omega t(F,G_i) = t(F,\lim_\omega G_i)$. Observe that all of the above notation makes perfect sense and the identities hold true for directed colored $r$-graphs, that is, when the adjacency arrays of the $G^\alpha$'s are not necessarily symmetric.

We call a measurable subset of $[0,1]^{\hhh([r])}$ an $r$-set graphon satisfying the usual symmetries in the coordinates induced by $S_r$ permutations, we can turn it into a proper $r$-graphon in the sense of \Cref{ch6:sec.prel} by generating the marginal with respect to the coordinate corresponding to $[r]$. Analogously a $k$-colored $r$-set graphon is a measurable partition of $[0,1]^{\hhh([r])}$ into $k$ classes invariant under coordinate permutations induced by permuting $[r]$. These objects can serve as representations of the ultralimits of $r$-graph sequences in the sense that the numerical sequences of subgraph densities converge to densities defined for $r$-set graphons in accordance with the notation in \Cref{ch6:sec.prel}, we will provide the definition next.

\begin{definition}
Let $\FFF$ be a $k$-colored $r$-graph on $m$ vertices, and $\WW=(W^1, \dots, W^k)$ be a $k$-colored $r$-set graphon. Then $T(\FFF,\WW) \subset [0,1]^{\hhh([m],r)}$ denotes the set of the symmetric maps $g:\hhh([m],r) \to [0,1]$ that satisfy that for each $e \in {[m] \choose r}$ it holds that  $(g(f))_{f \in \hhh(e)} \in W^{\FFF(e)}$. For the Lebesgue measure of $T(\FFF,\WW)$ we write $t(\FFF,\WW)$, this expression is referred to as the density of $\FFF$ in $\WW$. 
\end{definition} 

The reader may easily verify that the above definition of density agrees with the content of \Cref{ch6:sec.prel}.
One of the main technical results of \cite{ESz} is the following.

\begin{theorem}\label{ch6:thm.sep}
\cite{ESz} Let $r$ be an arbitrary positive integer and let $\mathcal A$ be a separable sub-$\sigma$-algebra of $\mathcal B_{[r]}$. Then there exists a separable realization $\phi: \XX^{[r]} \to [0,1]^{\hhh([r])}$ such that for every $A \in \mathcal A$ there exists a measurable $B \subset [0,1]^{\hhh([r])}$ such that $\mu_{[r]}(A \triangle \phi^{-1}(B))=0$. 

\end{theorem}

A lifting of a separable realization $\phi: \XX^{[r]} \to [0,1]^{\hhh([r])}$ of degree $n$ for $n \geq r$ is a measure preserving map $\psi:  \XX^{[n]} \to [0,1]^{\hhh([n],r)}$ that satisfies $p_{\hhh([r])} \circ \psi=\phi \circ P_{[r]}$, and it is equivariant under coordinate permutations in $S_n$, where $p_{\hhh([r])}$ and $P_{[r]}$ are the natural projections from $[0,1]^{\hhh([n],r)}$ to $[0,1]^{\hhh([r])}$, and from $\XX^{[n]}$ to $\XX^{[r]}$ respectively. The next lemma is central to relate the sub-$r$-graph densities of ultraproducts to the corresponding densities in $r$-set graphons.

%

\begin{lemma}\cite{ESz} \label{ch6:lem.lift}
For every separable realization $\phi$ and integer $n \geq r $ there exists a degree $n$ lifting $\psi$. 

\end{lemma}

The next statement is the colored version of the homomorphism correspondence in \cite{ESz} (Lemma 3.3. in that paper).

\begin{lemma}\label{ch6:lem.hom}
Let $\phi$ be a separable realization and $\WW=(W^1, \dots, W^k)$ be a $k$-colored $r$-graphon, and let $\HHH=(\HHH^1, \dots, \HHH^k)$ be a $k$-colored ultraproduct with $\mu_{[r]}(\HHH^\alpha \triangle \phi^{-1}(W^\alpha))=0$ for each $\alpha \in [k]$. Let $\psi$ be a degree $m$ lifting of $\phi$ and $\FFF$ be a $k$-colored $r$-graph on $m$ vertices. Then $\mu_{[m]}(\psi^{-1}(T(\FFF,\WW)) \triangle T(\FFF,\HHH))=0$, and consequently $t(\FFF,\WW)=t(\FFF,\HHH)$ for each $\FFF$.

\end{lemma}

\begin{proof}
By definition we have that $$
T(F,\HHH)= \bigcap_{e \in {[m] \choose r}} P_e^{-1} (P_{s_e} (\HHH^{F(e)}))$$
and
$$
T(F,W)= \bigcap_{e \in {[m] \choose r}} p_{\hhh([r])}^{-1} (p_{s_e} (W^{F(e)})).
$$
Due to the fact that $\psi$ commutes with coordinate permutations from $S_n$ and the conditions we imposed on the symmetric difference of $\HHH^\alpha$ and $\phi^{-1}(W^\alpha)$ the statement follows.
\end{proof}

We turn to describe the relationship of two $r$-set graphons whose $\FFF$-densities coincide for each $\FFF$. For this purpose we have to introduce two types of transformations and clarify their connection. Let us define the $\sigma$-algebras $\mathcal A_S$, $\mathcal A^*_S$, and $\mathcal B_S \subset\LL_{[0,1]^{\hhh([r])} }$ for each $S \subset [r]$, the $\sigma$-algebra $\mathcal B_S = p^{-1}_S( \LL_{[0,1]^{\hhh(S)} } )$, $\mathcal A_S$ is $\langle \mathcal B_T | T \subset S \rangle$, and $\mathcal A^*_S$ is $\langle \mathcal B_T | T \subset S, T \neq S \rangle$, where $\LL_{[0,1]^t}$ denotes the Lebesgue measurable subsets of the unit cube with the dimension given by the index.

\begin{definition}\label{ch6:defstrpres}
We say that the measurable map $\phi:[0,1]^{\hhh([r])} \to [0,1]^{\hhh([r])}$ is structure preserving if it is measure preserving, for any $S \subset [r]$ we have $\phi^{-1}(\mathcal A_S) \subset \mathcal A_S$, for any measurable $I \subset [0,1]$ we have $\phi^{-1}(p_S^{-1} (I))$ is independent of $\mathcal A^*_S$, and for any $\pi \in S_r$ we have $\Pi \circ \phi= \phi \circ \Pi$, where $\Pi$ is the coordinate permutation action induced by $\pi$.
\end{definition}

Let $\mathcal L^{\hhh([r])}$ denote the measure algebra of $([0,1]^{\hhh([r])},\LL_{[0,1]^{\hhh([r])} }, \lambda )$. 

\begin{definition}
We call an injective homomorphism $\Phi:\mathcal L^{\hhh([r])} \to \mathcal L^{\hhh([r])}$ a structure preserving embedding if it is measure preserving, for any $S \subset [r]$ we have $\Phi(\mathcal B_S) \subset \mathcal A_S$, $\Phi(\mathcal B_S)$ is independent from $\mathcal A^*_S$, and for any $\pi \in S_r$ we have $\Pi \circ \Phi= \Phi \circ \Pi$.
\end{definition}

Another result from \cite{ESz} sheds light on the build-up of structure preserving embeddings. 

\begin{lemma}\cite{ESz}\label{ch6:lem.rep}
Suppose that $\Phi:\mathcal L^{\hhh([r])} \to \mathcal L^{\hhh([r])}$ is a structure preserving embedding of a measure algebra into itself. Then there exists a structure preserving map $\phi:[0,1]^{\hhh([r])} \to [0,1]^{\hhh([r])}$ that represents $\Phi$ in the sense that for each $[U] \in L^{\hhh([r])}$ it holds that $\Phi([U])=[\phi^{-1}(U)]$, where $U$ is a representative of $[U]$.
\end{lemma}

A random coordinate system $\tau$ is the ultraproduct function on $\XX^{[r]}$ of the random symmetric functions $\tau_n: [n]^r \to [0,1]^{\hhh([n],r)}$ that are for each $n$ given by a uniform random point $Z_n$ in $[0,1]^{\hhh([n],r)}$ so that $(\tau_n(i_1, \dots, i_r))_e=(Z_n)_{p_e(i_1, \dots, i_r)}$. An important property of the random mapping $\tau_n$ is that for any $r$-set graphon and positive integer $n$ it holds that $(\tau_n)^{-1}(U)=\G(n,U)$, when the random sample $Z_n$ used to generate the two objects is the same.

\begin{lemma}\cite{ESz}
Let $U$ be an $r$-set graphon, and let $\HHH=[\{\G(n,U)\}_{n=1}^\infty]$. Then the random coordinate system $\tau=[\{\tau_n\}_{n =1}^\infty]$ is a separable realization such that with probability one we have $\mu_{[r]}(\HHH \triangle \tau^{-1}(U))=0$.
\end{lemma}

A direct consequence is the statement for $k$-colored $r$-set graphons.
\begin{corollary}\label{ch6:cor.sep}
Let $\UU=(U^1, \dots, U^k)$ be a $k$-colored $r$-set graphon, and let $\HHH=(\HHH^1, \dots, \HHH^k)$ be a $k$-colored ultraproduct in $\XX^{[r]}$, where $\HHH^\alpha=[\{\G(n,U^\alpha)\}_{n=1}^\infty]$ for each $\alpha \in [k]$. Then a random separable realization $\tau$ is such that with probability one we have $\mu_{[r]}(\HHH^\alpha \triangle \tau^{-1}(U^\alpha))=0$ for each $\alpha \in [k]$.
\end{corollary}

The following result is a generalization of the uniqueness assertion of \cite{ESz}, and states that subgraph densities determine an $r$-set graphon up to structure preserving transformations.
\begin{theorem}\label{ch6:thm.uni}
Let $\UU=(U^1, \dots, U^k)$ and $\VV=(V^1, \dots, V^k)$ be two $k$-colored $r$-set graphons such that for each $k$-colored $r$-graph $\FFF$ it holds that $t(\FFF,\UU)=t(\FFF,\VV)$. Then there exist two structure preserving maps $\nu_1$ and $\nu_2$ from $[0,1]^{\hhh([r])}$ to $[0,1]^{\hhh([r])}$ such that $\mu_{[r]}(\nu_1^{-1}(U^\alpha)\triangle \nu_2^{-1}(V^\alpha))=0$ for each $\alpha \in [k]$.

\end{theorem}
\begin{proof}
The equality $t(\FFF,\UU)=t(\FFF,\VV)$ for each $\FFF$ implies that $\GG(n, \UU)$ and $\GG(n,\VV)$ have the same distribution $Y_n$ for each $n$. Let $\HHH=[\{Y_n\}_{n=1}^\infty]$, then \Cref{ch6:cor.sep} implies that there exist separable realizations $\phi_1$ and $\phi_2$ such that $\mu_{[r]}(\HHH^\alpha \triangle \phi_1^{-1}(U^\alpha))=0$ and $\mu_{[r]}(\HHH^\alpha \triangle \phi_2^{-1}(V^\alpha))=0$ for each $\alpha \in [k]$, therefore also $\mu_{[r]}(\phi_1^{-1}(U^\alpha) \triangle \phi_2^{-1}(V^\alpha))=0$. 
Set $\mathcal A= \sigma(\phi_1^{-1}(\LL_{[0,1]^{\hhh([r])} }),\phi_1^{-1}(\LL_{[0,1]^{\hhh([r])} }) )$ that is a separable $\sigma$-algebra on $\XX^{[r]}$ so by \Cref{ch6:thm.sep} there exists a separable realization $\phi_3$ such that for each measurable 
$A \subset [0,1]^{\hhh([r])}$ the element $\phi_i^{-1}(A)$ of $\mathcal A$ can be represented by a subset  of $[0,1]^{\hhh([r])}$ denoted by $\psi_i(A)$. It is easy to check that the maps $\psi_1$ and $\psi_2$ defined this way are structure preserving embeddings from $\mathcal L^{\hhh([r])} \to \mathcal L^{\hhh([r])}$ satisfying $\lambda(\psi_1(U^\alpha)\triangle \psi_2(V^\alpha))=0$ for each $\alpha \in [k]$. We conclude that by 
\Cref{ch6:lem.rep} there are structure preserving $\nu_1$ and $\nu_2$  such that $\lambda(\nu_1^{-1}(U^\alpha)\triangle \nu_2^{-1}(V^\alpha))=0$ for each $\alpha \in [k]$.
\end{proof}

The next result is perhaps also meaningful beyond the framework of this paper and is the main contribution in the current section. Recall the definition of the ground state energies (GSE), \Cref{ch6:def.gse1} and \Cref{ch6:def.gse2}. 

\begin{theorem}\label{ch5:samp}
For any $J=(J^1, \dots, J^k)$ with $J^{\alpha}$ being a real $r$-array of size $q$ for each $\alpha \in [k]$ the parameter of $k$-colored $r$-graphs $\EEE_{r-1} ( . ,J)$ is testable.
\end{theorem}

\begin{proof}

We may assume that $\|J^\alpha\|_\infty \leq 1$ for every $\alpha$ without losing generality.
We proceed by contradiction. Suppose there exist an $\varepsilon>0$ and a sequence of $k$-colored $r$-uniform hypergraphs $\{\HHH_n\}_{n=1}^\infty$ with $V(\HHH_n)=[m_n]$ tending to infinity that are such that for each $n$ with probability at least $\varepsilon$ we have that $\EEE_{r-1} (\HHH_n,J) + \varepsilon \leq \EEE_{r-1} (\G(n, \HHH_n),J)$. Let $\GG_n=(G_n^1, \dots, G_n^k)$ denote the random $k$-colored hypergraph $\G(n, \HHH_n)$ for each $n$ with $G_n^\alpha=\G(n, H_n^\alpha)$.
The previous event can be reformulated as stating that for each $n$ with probability at least $\varepsilon$ there is a partition $\P_n=(P_n^1, \dots, P_n^q)$ of ${[n] \choose r-1}$ such that the expression
\begin{align*}
 \frac{1}{n^r}\sum_{\alpha=1}^{k} \sum_{i_1, \dots, i_r =1}^q J^\alpha(i_1, \dots, i_r) e_{G_n^\alpha} (r; P_n^{i_1}, \dots, P_n^{i_r})\end{align*}
 is larger than 
 \begin{align*} \frac{1}{m_n^r}\sum_{\alpha=1}^{k} \sum_{i_1, \dots, i_r =1}^q J^\alpha(i_1, \dots, i_r) e_{H_n^\alpha} (r; R_n^{i_1}, \dots, R_n^{i_r}) + \varepsilon \end{align*}
for any partition $\RRR_n=(R_n^1, \dots, R_n^q)$ of ${[m_n] \choose r-1}$.

Let $\HHH$ denote the ultralimit of the hypergraph sequence $\{\HHH_n\}_{n=1}^\infty$ that is a $k$-partition in the measure space $(\XX^{[r]}_1, \mathcal B_1, \mu_1)$, and let $\sigma_1(S)$ and $\sigma_1(S)^*$ denote the sub-$\sigma$-algebras of $\mathcal B_1$ corresponding to subsets $S$ of $[r]$. Due to \Cref{ch6:thm.sep} 
 there exists a separable realization $\phi_1: \XX^{[r]}_1 \to [0,1]^{\hhh([r])}$ such that there is a $k$-colored $r$-set graphon $\WW=(W^1, \dots, W^k)$ satisfying $\mu_1(\phi_1^{-1}(W^\alpha) \triangle \HHH^\alpha)=0$ for each $\alpha \in [k]$. Let $\GG(s)$ stand for the point-wise ultralimit realization of the $\{\GG_n(s)\}_{n=1}^\infty \subset \XX^{[r]}_2$ for all $s \in \mathbb S$, where $(\mathbb S, \mathcal S, \nu)$ denotes the underlying joint probability space for the random hypergraphs, and $(\XX^{[r]}_2, \mathcal B_2, \mu_2)$ is the ultraproduct measure space in the case of the sample sequence, $\sigma_2(S)$ and $\sigma_2(S)^*$ are the corresponding sub-$\sigma$-algebras. Note that the ultralimits $\GG(s)$ are not $k$-partitions of the same ultraproduct space as $\HHH$, moreover, it is possible that the $\sigma$-algebra generated by $\{\GG(s)| s \in \mathbb S\}$ together with $\mu_2$ form a non-separable measure algebra that prevents us from using \Cref{ch6:thm.sep}
 directly. 

Suppose that for some $n$ we have that $\E \EEE_{r-1} (\GG_n,J) < \EEE_{r-1} (\HHH_n,J) + 3/4 \varepsilon$. This assumption implies by \Cref{ch5:conc} that $\PPP(\EEE_{r-1} (\GG_n,J) \geq \EEE_{r-1} (\HHH_n,J) + \varepsilon ) \leq \PPP(\EEE_{r-1} (\GG_n,J) \geq \E \EEE_{r-1} (\GG_n,J) + \varepsilon/4) \leq 2\exp(-\frac{\varepsilon^2 n}{64 r^2})$. The last bound is strictly smaller than $\varepsilon$ when $n$ is chosen sufficiently large, therefore it contradicts the main assumption for large $n$. Therefore we can argue that $\E \EEE_{r-1} (\GG_n,J) \geq \EEE_{r-1} (\HHH_n,J) + 3/4 \varepsilon$ for large $n$, throwing away a starting piece of the sequence $\{\HHH_n\}_{n=1}^\infty$ we may assume that it holds for all $n$. 

A second application of \Cref{ch5:conc} leads to a lower bound on the probability that $\EEE_{r-1} (\GG_n,J)$ is close to $\EEE_{r-1} (\HHH_n,J)$, namely $\PPP(\EEE_{r-1} (\GG_n,J) \leq \EEE_{r-1} (\HHH_n,J) + \varepsilon/2) \leq 2\exp(-\frac{\varepsilon^2 n}{64 r^2})$. Hence, by invoking the Borel-Cantelli Lemma, we infer that with probability one the event $\EEE_{r-1} (\GG_n,J) \leq \EEE_{r-1} (\HHH_n,J) + \varepsilon/2$ can occur only for finitely many $n$, let the $M_1$ denote the (random) threshold for which is true that  $\EEE_{r-1} (\GG_n,J) > \EEE_{r-1} (\HHH_n,J) + \varepsilon/2$ for every $n \geq M_1$. It follows that $\lim_\omega \EEE_{r-1} (\GG_n,J) > \lim _\omega \EEE_{r-1} (\HHH_n,J) + \varepsilon/2$  with probability $1$.    

Next we will show that with probability one $\GG$ is equivalent to $\HHH$ in the sense that for each $k$-colored $r$-graph $\FFF$ it holds that $t(\FFF,\GG)=t(\FFF,\HHH)$. Then, since there are countably many test graphs $\FFF$, we can conclude that the equality holds simultaneously for all $\FFF$ with probability $1$. 

We have seen above in the paragraph after \eqref{ch6:eq2}  that for every fixed $k$-colored $r$-uniform hypergraph $t(\FFF,\HHH)=\lim_{\omega} t(\FFF,\HHH_n)$. On the other hand the subgraph densities in random induced subgraphs are highly concentrated around their mean, that is $$\PPP(|t(\FFF, \GG_n)-t(\FFF,\HHH_n)| \geq \delta) \leq 2 \exp(-\frac{\delta^2 n}{2|V(\FFF)|^2})$$ for any $\delta>0$, this follows with basic martingale techniques, see Theorem 11 in \cite{ESz} for the almost identical statement together with a proof. The Borel-Cantelli Lemma implies then for every fixed $\FFF$ that with probability one for each $\delta>0$ there exists a (random) $n_0(\delta)$ such that for each $n \geq n_0(\delta)$ it is true that $|t(\FFF,\GG_n)-t(\FFF,\HHH_n)| < \delta/2$. Let us fix $\delta >0$ and $\FFF \in \GGG^{r,k}$. Since the set $\{n |\quad |t(F,\HHH_n)-t(F,\HHH)| < \delta/2 \}$ belongs to  $\omega$ by the definition of the ultralimit function, it holds that  $\{n | \quad |t(\FFF,\GG_n)-t(F,\HHH)| < \delta \} \in \omega$ as a consequence of 
\begin{align*}
\{n | &\quad |t(\FFF,\GG_n)-t(\FFF,\HHH)| < \delta \} \\ &\supset (\{n | \quad |t(\FFF,\GG_n)-t(\FFF,\HHH_n)| < \delta/2 \} \cap \{n | \quad |t(\FFF,\HHH_n)-t(\FFF,\HHH)| < \delta/2 \}) \in \omega. 
\end{align*}
 Consequently, $\lim_\omega t(\FFF,\GG_n)=t(\FFF,\HHH)$ with probability one for each $\FFF$, and the limit equation holds simultaneously for each $\FFF$ also with probability one, since their number is countable.

Let us pick a realization $\{\GG_n(s)\}_{n=1}^\infty$ of $\{\GG_n\}_{n=1}^\infty$ such that it satisfies $\lim_{\omega} \EEE_{r-1} (\GG_n(s),J) - \lim_\omega \EEE_{r-1} (\HHH_n,J) \geq \varepsilon/2$ and $\lim_\omega t(\FFF,\GG_n(s))=t(\FFF,\HHH)$ for each $\FFF$, the preceding discussion implies that such a realization exists, in fact almost all of them are like this. Furthermore, let us consider the sequence of partitions $\P_n=(P_n^1, \dots, P_n^q)$ of ${[n] \choose r-1}$ that realize $\EEE_{r-1} (\GG_n(s),J)$, and define $T^{i,j}_n\subset [n]^r \setminus \diag([n]^r) $ through the inverse images of the projections $A_{T^{i,j}_n}=(p^n_j)^{-1}(A_{P_n^i})$ for $i \in [q]$, $j \in [r]$, and $n \in \N$, where $p^n_j$ is the projection that maps an $r$-array of size $n$ onto an $(r-1)$-array by erasing the $j$th coordinate. Note that the $T^{i,j}_n$'s are not completely symmetric, but are invariant under coordinate permutations from $S_{[r] \setminus \{j\}}$ for the corresponding $j \in [r]$. A further property is that and  $T^{i,j_1}_n$ can be obtained from $T^{i,j_2}_n$ swapping the coordinates corresponding to $j_1$ and $j_2$.

We additionally define the ultraproducts of these sets by  $\PP^i=[\{P^i_n\}_{n=1}^\infty] \subset \XX_2^{[r-1]}$ and $\TT^{i,j}=[\{T^{i,j}_n\}_{n=1}^\infty] \subset \XX_2^{[r]}$, it is clear that $\TT^{i,j} \in \sigma_2([r]\setminus \{j\})$ for each pair of $i$ and $j$, so $\cap_{(i,j)\in I} \TT^{i,j} \in \sigma_2([r])^*$ for any $I \subset [q]\times [r]$, and that $\XX_2^{[r-1]} = \cup_i \PP^i$. The same symmetry assumptions apply for the $\TT^{i,j}$'s as for the $T^{i,j}_n$'s described above.

We also require the fact that these ultraproduct sets defined above establish a correspondence between the GSE of $\GG(s)$ and the ultralimit of the sequence of energies $\{\EEE_{r-1} (\GG_n(s),J)\}_{n=1}^\infty$. 

This  can be seen as follows: Recall that  
\begin{align*}
 \EEE_{r-1} (\GG_n(s),J) = \frac{1}{n^r}\sum_{\alpha=1}^{k} \sum_{i_1, \dots, i_r =1}^q J^\alpha(i_1, \dots, i_r) |G_n^\alpha\cap (\cap_{j=1}^q T_n^{i_j,j})|\end{align*}
This formula together with the identities $[\{G^\alpha_n(s) \cap (\cap_{j=1}^q T_n^{i_j})\}_{n=1}^\infty] = \GG^\alpha(s) \cap (\cap_{j=1}^q \TT^{i_j,j})$, and 
that the ultralimit of subgraph densities equals the subgraph density of the ultraproduct imply that
\begin{align*}
\lim_\omega \EEE_{r-1} (\GG_n(s),J) = \sum_{\alpha=1}^{k} \sum_{i_1, \dots, i_r =1}^q J^\alpha(i_1, \dots, i_r) \mu_2(\GG^\alpha(s) \cap (\cap_{j=1}^q \TT^{i_j,j})).
\end{align*}
Now consider the separable sub-$\sigma$-algebra $\mathcal A$ of $\mathcal B_2$ generated by the collection of the sets $\GG^1(s), \dots, \GG^k(s), \TT^{1,1}, \dots, \TT^{q,r}$. Then by \Cref{ch6:thm.sep} 
 there exists a separable realization $\phi_2: \XX_2^{[r]} \to [0,1]^{\hhh([r])}$ and measurable sets $U^1, \dots, U^k, V^{1,1}, \dots, V^{q,r}$ such that $\mu_{2}(\phi_2^{-1}(U^\alpha) \triangle \GG^\alpha(s))=0$  for each $\alpha \in [k]$ and $\mu_{2}(\phi_2^{-1}(V^{i,j}) \triangle \TT^{i,j})=0$ for every $i \in [q],j \in [r]$. Additionally, we can modify the $V^{i,j}$'s on a set of measure $0$ such that each of them only depends on the coordinates corresponding to the sets in $\hhh([r]\setminus \{j\})$, is invariant under coordinate permutations induced by elements of $S_{[r]}$ that fix $j$, and  $V^{i,j_1}$ can be obtained from $V^{i,j_2}$ by relabeling the coordinates according to the $S_r$ permutation swapping $j_1$ and $j_2$. Also, $(U^1, \dots, U^k)$ form a $k$-colored $r$-set graphon $\UU$ when we make modifications on null sets. Most importantly, the separable realization $\phi_2$ is measure preserving, so we have that
\begin{equation} \label{ch6:eq20}
\lim_\omega \EEE_{r-1} (\GG_n(s),J) = \sum_{\alpha=1}^{k} \sum_{i_1, \dots, i_r =1}^q J^\alpha(i_1, \dots, i_r) \lambda(U^\alpha \cap (\cap_{j=1}^r V^{i_j,j})).
\end{equation}
On the other hand we established that $t(\FFF,\GG(s))=t(\FFF,\HHH)$ for each $\FFF$,which implies $t(\FFF,\UU)=t(\FFF,\WW)$, therefore the uniqueness statement of \Cref{ch6:thm.uni}
 ensures the existence of two structure preserving measurable maps $\nu_1, \nu_2:[0,1]^{\hhh([r])} \to [0,1]^{\hhh([r])}$ such that $\lambda(\nu_1^{-1}(W^\alpha) \triangle \nu_2^{-1}(U^\alpha))=0$ for each $\alpha \in [k]$. 

Now let us define the sets $\SSS^{i,j}=\phi_1^{-1}(\nu_2(\nu_1^{-1} (V^{i,j})))$, these satisfy exactly the same symmetry properties as the $\TT^{i,j}$'s above, by the measure preserving nature of the maps involved we have that 
\begin{equation}\label{ch6:eq21}
\lim_\omega \EEE_{r-1} (\GG_n(s),J) = \sum_{\alpha=1}^{k} \sum_{i_1, \dots, i_r =1}^q J^\alpha(i_1, \dots, i_r) \mu_1(\HHH^\alpha \cap (\cap_{j=1}^r \SSS^{i_j,j})).\end{equation}
 The properties of structure preserving maps imply that $\SSS^{i,j} \in \sigma_1([r]\setminus \{j\})$ for each $i,j$, so $\cap_{(i,j)\in I} \SSS^{i,j} \in \sigma_1([r])^*$ for any $I \subset [q]\times [r]$. Also, the ultraproduct construction makes it possible to assert the existence of a sequence of partitions $\RRR_n=(R_n^1, \dots, R_n^q)$ of ${[m_n] \choose {r-1}}$  for $\omega$-almost every $n$ such that $\SSS^{i,j}=[\{(p^{m_n}_j)^{-1} (R_n^i)\}_{n=1}^\infty]$. But again by the correspondence principle between ultralimits of sequences and ultraproducts in \Cref{ch6:lem.hom} applied to (\ref{ch6:eq20}) and (\ref{ch6:eq21}) we have  
$$\lim_\omega \EEE_{\RRR_n,r-1} (\HHH_n,J) = \lim_\omega \EEE_{r-1} (\GG_n(s),J),$$
which contradicts $\lim_{\omega} \EEE_{r-1} (\GG_n(s),J) - \lim_{\omega} \EEE_{r-1} (\HHH_n,J) \geq \varepsilon/2$.
      
\end{proof}

An immediate consequence is that the above theorem is also true for $r$-graphons.

\begin{corollary}\label{ch6:cor.samp}
For any $J=(J^1, \dots, J^k)$ with $J^{\alpha}$ being a real $r$-array of size $l$ for each $\alpha \in [k]$ there exists for any $\varepsilon>0$ a $q(\varepsilon)$ integer such that for any $k$-colored $r$-set graphon $\WW$ and $q \geq q(\varepsilon)$ it holds that $$ \PPP (|\EEE_{r-1} ( \WW,J)-\EEE_{r-1} (\G(q,\WW),J)| > \varepsilon ) < \varepsilon. $$
\end{corollary}

\begin{proof}
We only sketch the proof here, details are left to the reader. The main idea is to find for any fixed $\varepsilon>0$, and  for each $k$-colored $r$-set graphon $\WW$ a $\GG \in \GGG^{r,k}$ such that their GSE are sufficiently close, and further, the distributions of $\G(q_0(\varepsilon/2), \WW)$ and $\G(q_0(\varepsilon/2),\GG)$ are close enough in terms of $\varepsilon$, where $q_0$ is the sample complexity of $\EEE_{r-1} (.,J)$,  whose existence is ensured by \Cref{ch5:samp}.  
Fix $\varepsilon>0$, and let $\WW$ be a $k$-colored $r$-set graphon. By measurability for any $\Delta >0$ there exists an integer $l$ and a $k$-colored $r$-set graphon $\UU$ such that each $U^\alpha$ is a union of cubes $\times_{S \in \hhh([r])} [\frac{z_S-1}{l}, \frac{z_S}{l}]$ with $z \in \Z^{\hhh([r])}$ and $\sum_{\alpha=1}^k \|W^\alpha-U^\alpha\|_1 \leq \Delta$. For a fixed, but sufficiently small $\Delta$, let $\GG$ be the $k$-colored $r$-graph on $l$ vertices that is obtained by randomization form $\UU$ using the independent uniform $[0,1]$-valued random variables $(X_S)_{S \in \hhh([l],r)\setminus \hhh([l],1)}$. Then by standard large deviations results it follows that the $1$-norm of $U^\alpha-W_{\GG^\alpha}$ is highly concentrated around $0$. By definition, the deviation of the GSE's of two $r$-graphons can be upper bounded by a constant multiple of their difference in the $1$-norm. By \Cref{ch6:hypercountcor} the same is true for the total variation distance of the corresponding measures for the fixed sampling depth $q_0(\varepsilon/2)$, as the cut-norm is dominated by the $1$-norm. The quantity $\sum_{\alpha=1}^k \|W^\alpha-W_{\GG^\alpha}\|_1$ can be made arbitrarily small by the above discussion, which proves the result.   
\end{proof}
We can derive a substantial property of the cut norm form the above theorem. Recall the definition of the relevant norms, \Cref{ch6:def.cutnorm}.

\begin{lemma}\label{ch6:hyperregpres}
Let $r \geq 1$. For any $\varepsilon >0 $ and $t \geq 1$ there exists an integer $l_{0}(r,\varepsilon,t)\geq 1$ a such that  for any symmetric $r$-kernel $U$ that takes values in $[-1, 1]$, and  for any integer $l \geq l_0(r,\varepsilon,t)$ it holds with probability at least $1- \varepsilon$ that
$$
\left| \sup_{\QQ, t_\QQ\leq t} \|U\|_{\square,r, \QQ} - \sup_{\QQ, t_\QQ\leq t} \|W_{\G(l,U)}\|_{\square,r, \QQ} \right| \leq \varepsilon,
$$ 
 where the supremum at both places goes over symmetric partitions $\QQ$ of $[0,1]^{\hhh([r-1])}$ into at most $t$ classes.

\end{lemma}

\begin{proof}
Let us fix $\varepsilon >0$, $r,t \geq 1$, and let $U$ be arbitrary. In this lemma we deal with $r$-graphons instead of $r$-set graphons, Fubini's Theorem ensures that we can apply \Cref{ch5:samp} correctly later on. 

Showing that there exits an $l_0$ not depending on $U$ such that for each $l \geq l_0$ it holds that 
$\sup_{\QQ, t_\QQ\leq t}\|U\|_{\square,r, \QQ} - \sup_{\QQ, t_\QQ\leq t}\|W_{\G(l,U)}\|_{\square,r, \QQ}  \leq \varepsilon$ with failure probability at most $\varepsilon /2$ is a routine exercise, we only have to consider a tuple $(S_i)_{i \in [r]}$ of symmetric subsets of $[0,1]^{\hhh([r-1])}$ and a symmetric partition $\QQ^0$ of $[0,1]^{\hhh([r-1])}$ into at most $t$ classes such that 
\begin{align*}\sup_{\QQ, t_\QQ\leq t}\|U\|_{\square,r, \QQ}=\sum_{j_1, \dots,j_r=1}^t |\int_{\cap_{i \in [r]} p_{[r]\setminus \{i\}}^{-1}(S_i \cap Q^0_{j_i} ) } U(x_{\hhh([r],r-1)}) \du \lambda(x_{\hhh([r],r-1)}) |,
\end{align*}
and use Markov's inequality.
The difficult part is to show that if $l$ is large enough then  for each $U$ it holds that
$$
\sup_{\QQ, t_\QQ\leq t} \|W_{\G(l,U)}\|_{\square,r, \QQ}- \sup_{\QQ, t_\QQ\leq t} \|U\|_{\square,r, \QQ} \leq \varepsilon
$$
with probability at least $1- \varepsilon/2$.

First we have to discretize the range of $U$ in order to apply the above result on $k$-colored $r$-graphs, \Cref{ch6:cor.samp}. Therefore we split the interval $[-1,1]$ into consecutive intervals $I_1, \dots, I_k$ of length at most $\varepsilon/4$, let $y_j=\inf I_j$ for each $j \in [k]$, and define the $r$-kernel $W(x)=\sum_{j=1}^k \I_{I_j} (U(x)) y_j$. Then $\|U-W\|_\infty \leq \varepsilon/4$, so therefore $\left|\|U\|_{\square,r, \QQ}-\|W\|_{\square,r, \QQ}\right| \leq \varepsilon/4$ and $\left|\|W_{\G(l,U)}\|_{\square,r, \QQ}-\|W_{\G(l,W)}\|_{\square,r, \QQ}\right| \leq \varepsilon/4$ for any $\QQ$ and $l$. Thus, it suffices to 
show the existence of an $l_0$ not depending on $U$ or $W$ such that for each $l \geq l_0$ we have $$
 \|W_{\G(l,W)}\|_{\square,r, \QQ} - \|W\|_{\square,r, \QQ}  \leq \varepsilon/2
$$ 
for each partition symmetric $\QQ$ of $[0,1]^{\hhh([r-1])}$ into at most $t$ classes simultaneously with probability at least $1- \varepsilon/2$.

We can rewrite $\sup_{\QQ, t_\QQ\leq t} \|W\|_{\square,r, \QQ}$ as an optimization problem, more precisely 
\begin{align}\label{ch6:eq14}
&\sup_{\QQ, t_\QQ\leq t} \|W\|_{\square,r, \QQ} \\ & =
\sup_{\QQ, t_\QQ\leq t}\max_{A \in \mathbb A}\sup_{\substack{T_j \subset [0,1]^{\hhh([r-1])}  \\ j \in [r]}}
\sum_{i_1, \dots, i_r=1}^t A(i_1, \dots, i_r) \int\limits_{[0,1]^{\hhh([r],r-1)}} W(x_{\hhh([r],r-1)}) \prod_{j=1}^r \I_{T_j \cap Q_{i_j}} (x_{\hhh([r] \setminus \{j\})})    \du \lambda(x_{\hhh([r],r-1)}),
\end{align}
where $\mathbb A$ denotes the set of all $r$-arrays of size $t$ with $\{-1,1\}$ entries, and the set and partitions involved are symmetric.

If we swap the order of the maximization operation on the right of the above formula (\ref{ch6:eq14}), then it can be turned into a generalized energy for each $A \in \mathbb A$. In more detail, consider $W$ as a $k$-colored $r$-graphon with $W^{\alpha}= \I_{W=y_\alpha}$ for each $\alpha \in [k]$, with slight abuse of notation we set $W=(W^{\alpha})_{\alpha \in [k]}$.  We also define the $r$-array $B_0$ of size $2^r$, indexed by the power set of $[r]$ so that $B_0(i_{S_1}, \dots, i_{S_r})$ is equal to $1$ if for every $j \in [r]$ we have $j \in S_j$, and is equal to $0$ otherwise. Let $J^\alpha_A=y_\alpha (A \otimes B_0)$ be the tensor product of $A$ and $B_0$ for each $A \in \mathbb A$ multiplied with the scalar $y_\alpha$ with $\alpha \in [k]$, then $J_A^\alpha$ is an $r$-array of size $2^r t$. 
It follows that 
\begin{align}\label{ch6:eq13}
\max_{A \in \mathbb A} \EEE_{r-1}(W, J_A) = \sup_{\QQ, t_\QQ\leq t} \|W\|_{\square,r, \QQ}.
\end{align} 
Similarly, 
$$\max_{A \in \mathbb A} \EEE_{r-1}(W_{\G(l,W)}, J_A)=\sup_{\QQ, t_\QQ\leq t} \|W_{\G(l,W)}\|_{\square,r, \QQ},$$
hence 
\begin{align*}
\sup_{\QQ, t_\QQ\leq t} \|W_{\G(l,W)}\|_{\square,r, \QQ}- \sup_{\QQ, t_\QQ\leq t} \|W\|_{\square,r, \QQ} \leq  \max_{A \in \mathbb A} |\EEE_{r-1}(W_{\G(l,W)}, J_A) -\EEE_{r-1}(W, J_A)|. 
\end{align*}
The function $\EEE_{r-1}( . , J_A)$ is testable by \Cref{ch6:cor.samp}, say with sample complexity $q_1(\varepsilon,r,l,k)$, so $\sup_{\QQ, t_\QQ\leq t} \| . \|_{\square,r, \QQ}$ is testable with sample complexity $l_0(r,\varepsilon,t) = q_1(\varepsilon / |\mathbb A|,r,m, 2^r t)$. 

\end{proof}

In fact, we will require the version of \Cref{ch6:hyperregpres} for $k$-tuples $r$-kernels.

\begin{lemma}\label{ch6:hyperregpres2}
Let $r,k \geq 1$. For any $\varepsilon >0 $ and $t \geq 1$ there exists an integer $q_{\cut}(r,k,\varepsilon,t)\geq 1$ a such that  for any  $k$-tuple of $r$-kernel $U_1, \dots, U_k$ that take values from $[-1, 1]$, and any integer $q \geq q_{\cut}(r,k,\varepsilon,t)$ it holds with probability at least $1- \varepsilon$ that
$$
\left| \sup_{\QQ, t_\QQ\leq t} \sum_{j=1}^{k}\|U_j\|_{\square,r, \QQ} - \sup_{\QQ, t_\QQ\leq t} \sum_{j=1}^{k} \|W_{\G(q,U_j)}\|_{\square,r, \QQ} \right| \leq \varepsilon
$$ 
 where the supremum at both places goes over symmetric partitions $\QQ$ of $[0,1]^{\hhh([r-1])}$ into at most $t$ classes.
\end{lemma}
\begin{proof}
We only sketch the proof as it is almost identical to that of \Cref{ch6:hyperregpres}. Let $r,k,t \geq 1$ and $\varepsilon>0$ be fixed, and let $U_1, \dots, U_k$ and $q$ be arbitrary. The lower bound on $\sup_{\QQ, t_\QQ\leq t} \sum_{j=1}^{k} \|W_{\G(q,U_j)}\|_{\square,r, \QQ}$ can be obtained by the same argument as above using Markov's inequality. For the upper bound we again discretize to obtain the $r$-kernels $W_1, \dots, W_k$ with common range $\{y_i: \alpha \in [m]\}$ such that $\|U_j-W_j\|_\infty \leq \frac{\varepsilon}{4k}$ for each $j \in [k]$, hence $m=\frac{8k}{\varepsilon}$ will do. We associate to each $W_j$ and $m$-colored $r$-graphon $\WW_j$ as above and set $J_A^\alpha$ to $y_\alpha(A \otimes B_0)$, then 
\begin{align*}
\max_{A_1, \dots, A_k \in \mathbb A} \sup_{\QQ, t_\QQ\leq t}\sum_{j=1}^k \EEE_{\QQ,r-1}(\WW_j, J_{A_j}) = \sup_{\QQ, t_\QQ\leq t} \sum_{j=1}^k\|W_j\|_{\square,r, \QQ}.
\end{align*} 
Similarly, 
$$ \max_{A_1, \dots, A_k \in \mathbb A} \sup_{\QQ, t_\QQ\leq t} \sum_{j=1}^k \EEE_{\QQ,r-1}(\WW_{\G(q,\WW_j)}, J_{A_j})=\sup_{\QQ, t_\QQ\leq t} \sum_{j=1}^k\|W_{\G(q,W_j)}\|_{\square,r, \QQ}.$$
The testability of $\sup_{\QQ, t_\QQ\leq t}\sum_{j=1}^k \EEE_{\QQ,r-1}(\WW_j, J_{A_j})$ follows from \Cref{ch5:samp} with a slight modification of the argument for any fixed tuple $A_1, \dots, A_k \in \mathbb A$. As the cardinality of $\mathbb A$ does not depend on $\WW_1, \dots, \WW_k$ the statement of the lemma follows.
\end{proof}

\section{Auxiliary lemmas}\label{ch6:sec.aux}
We will require the version of Szemer\'edi's  Regularity Lemma adapted to the Hilbert space setting. Let us recall this variant. 

\begin{lemma}\cite{LSzreg} \label{hreg}
Let $\KK_1, \KK_2, \dots$ be arbitrary subsets of a Hilbert space $\HH$. Then for every $\varepsilon>0$ and $f \in \HH$ there is an $ m \leq \frac{1}{\varepsilon^2}$ and there are $f_i \in \KK_i$ and $\gamma_i \in \R$ ($1 \leq i\leq m$) such that for every $g \in \KK_{m+1}$ we have that 
\begin{align*}
|\langle g, f- \sum_{i=1}^m \gamma_i f_i \rangle| \leq \varepsilon \|f\| \|g\|.
\end{align*}   
\end{lemma}
 
 We start with the following intermediate version of the regularity lemma for edge $k$-colored $r$-graphons, the partition obtained here satisfies stronger conditions than those imposed by the Weak Regularity Lemma \cite{FK}, and weaker than by Szemer\'edi's original.  
\begin{lemma}\label{ch6:sregkhyper}

 For every $r \geq 1$, $\varepsilon>0$, $t \geq 1$, $k \geq 1$ and $k$-colored  $r$-graphon $\WW$ there exists a symmetric partition $\P=(P_1, \dots, P_m)$ of $[0,1]^{\hhh([r-1])}$ into $ m \leq (2t)^{(rk+1)^{4/\varepsilon^2}}=t_\reg(r,k,\varepsilon, t)$ parts and a symmetric $(r,r-1)$- step function $\VV \in \WWW^{r,k}$ with steps from $\P$, such that for any partition $\QQ$ of $[0,1]^{\hhh([r-1])}$ into at most $m t$ classes we have
\begin{align*}
d_{\square,r,\QQ}(\WW,\VV) \leq \varepsilon. 
\end{align*} 
\end{lemma}

\begin{proof}

Our lemma is a special case of \Cref{hreg}. We set $\HH$ to be the space of of $k$-tuples of real measurable functions on $[0,1]^{\hhh([r],r-1)}$ with the sum of the component-wise $L^2$-products as the inner product, this space contains
 $\WWW^{r,k}$.  Set $s(1)=1$ and $s(i+1)=s(i)(s(i)t+1)^{rk}$ for each $i \geq 1$ and let $\KK_i$ be the  set of $k$-tuples of indicator functions that are $(r,r-1)$ step functions with $s(i)$ number of symmetric steps and taking values from the set $\{-1,0,1\}$. Note that the elements of the $\KK_i$'s are not necessarily symmetric as functions, only their steps are required to be such. Further, observe that $s(i) \leq (2t)^{(rk+1)^i}$. Now apply \Cref{hreg} with the above setup for $\varepsilon/2$ and $\WW$ to obtain $\UU$ that satisfies all the conditions of the lemma except for symmetry, in particular
 \begin{align*}
 \sum_{\alpha=1}^k \|U^\alpha-W^\alpha\|_{\square,r,\P} < \varepsilon.
 \end{align*}
 Define $\VV$ with $V^\alpha(x_{\hhh([r],r-1)})=\frac{1}{r!}\sum_{\pi \in S_r} U^\alpha(x_{\pi(\hhh([r],r-1))})$. The symmetry of $\WW$ and the triangle inequality implies that $\VV$ is suitable, since
 \begin{align*}
 \|V^\alpha-W^\alpha\|_{\square,r,\P} \leq \frac{1}{r!}\sum_{\pi \in S_r} \|(U^\alpha)^\pi-(W^\alpha)^\pi\|_{\square,r,\P}=\|U^\alpha-W^\alpha\|_{\square,r,\P}
 \end{align*}
 for any $\P$ and $\alpha \in [k]$, and $U^\pi(x_{\hhh([r],r-1)})=U(x_{\pi(\hhh([r],r-1))})$.

\end{proof}

The next lemma is analogous to Lemma 3.2 from \cite{KM2}. It describes under what metric conditions a $k$-coloring of a $t$-colored graphon can be transfered to another one so that the two $tk$-colored graphons are close in a certain sense.  
For the sake if completeness we sketch the proof.

\begin{lemma}\label{ch6:hypercoloring}
Let $\varepsilon>0$, $\UU$ be a $t$-colored $r$-graphon that is an $(r,r-1)$-step function with steps $\P=(P_1, \dots, P_m)$ and $\VV$ be a $t$-colored $r$-graphon with $d_{\square,r, \P}(U, V) \leq \varepsilon$. For any $k \geq 1$ and $\hat \UU$ a $[t]\times [k]$-colored $r$-graphon that is an $(r,r-1)$-step function with steps from $\P$ such that $[\hat \UU, k]=\UU$ there exists a $k$-coloring of $\VV$ denoted by $\hat \VV$ so that 
\begin{equation*}
d_{\square,r, \P}(\hat \UU, \hat \VV) \leq k \varepsilon.
\end{equation*}
\end{lemma}

\begin{proof}
Fix $\varepsilon >0$, and let $\UU=(U^\alpha)_{\alpha \in [t]}$, $\VV=(V^\alpha)_{\alpha \in [t]}$ and $\hat \UU=(U^{\alpha,\beta})_{\alpha \in [t], \beta \in [k]}$ as in the statement of the lemma. Then $\sum_{\alpha=1}^t U^\alpha=1$ and $\sum_{\beta =1}^k U^{\alpha,\beta}=U^{\alpha}$ for each $\alpha \in [t]$. Let us define $\hat \VV=(V^{\alpha,\beta})_{\alpha \in [t], \beta \in [k]}$ that is a $k$-coloring of $\VV$. Set $V^{\alpha,\beta}=V^{\alpha} [\I_{U^{\alpha}=0} \frac{1}{k} + \I_{U^{\alpha}>0} \frac{U^{\alpha,\beta}}{U^{\alpha}}]$, it is easy to see that the factor on the right of the formula is a $(r,r-1)$-step function with steps $\P=(P_1, \dots, P_m)$. We estimate the deviation of each pair $U^{\alpha,\beta}$ and $V^{\alpha,\beta}$ from above in the $r$-cut norm, for this we fix the symmetric $S_1, \dots, S_r \subset [0,1]^{\hhh([r-1])}$. Then we have 
\begin{align*}
\left|\int_{\cap_{l \in [r]}p_l^{-1}(S_l)}  U^{\alpha,\beta}-V^{\alpha,\beta}\right|  
&\leq \sum_{\alpha_1,\dots, \alpha_r=1}^t \left|\int_{\cap_{l \in [r]}p_l^{-1}(S_l \cap P_{\alpha_l})}  U^{\alpha,\beta}-V^{\alpha,\beta} \right|\\ &= \sum_{\alpha_1,\dots, \alpha_r=1}^t \left|\int_{\cap_{l \in [r]}p_l^{-1}(S_l \cap P_{\alpha_l})}  (U^{\alpha}-V^{\alpha}) [\I_{U^{\alpha}=0} \frac{1}{k} + \I_{U^{\alpha}>0} \frac{U^{\alpha,\beta}}{U^{\alpha}}] \right| \\  &\leq \|U^{\alpha}-V^{\alpha}\|_{\square,r,\P}.
\end{align*}
Taking the maximum over all symmetric measurable $r$-tuples $S_1, \dots, S_r$ and summing up over all choices of $\alpha$ and $\beta$ delivers the upper bound we were after.
\end{proof}


\section{Proof of the main result}\label{ch6:sec.main}

The central tool in the main proof is the following lemma which can also be of independent interest. Informally it states that every coloring of a sampled $r$-graph can be transferred onto the graphon from which the graph was sampled from, such that another sampling procedure with a much smaller sample size cannot distinguish the two colored objects.  
\begin{lemma}\label{ch6:mainhyperlemma}
 For every $r \geq 1$, proximity parameter $\delta > 0$, palette sizes $t,k\geq 1$, sampling depth $q_0 \geq 1$ 
 there exists an integer $q_\tw=q_\tw(r,\delta,q_0, t,k) \geq 1$ such that for every $q \geq q_\tw$ the following holds. Let $\UU=(U^\alpha)_{\alpha \in [t]}$ be a $t$-colored $r$-graphon and let $V^\alpha$ denote $W_{\G(q,U^\alpha)}$ for each $\alpha \in [t]$, also let $\VV=(V^\alpha)_{\alpha \in [t]}$, so  $\WW_{\G(q,\UU)}=\VV$. Then with probability at least $1-\delta$ there  exist for every $k$-coloring $\hat \VV=(V^{\alpha,\beta})_{\alpha \in [t], \beta \in [k]}$ of $\VV$ a $k$-coloring $\hat \UU=(U^{\alpha,\beta})_{\alpha \in [t], \beta \in [k]}$  of $\UU=(U^\alpha)_{\alpha \in [t]}$ 
 such that we have that 

\begin{equation*}
d_\tw (\mu(q_0, \hat \WW), \mu(q_0, \hat \UU) \leq \delta.
\end{equation*}  
\end{lemma}
\begin{proof}
We proceed by induction with respect to $r$. The statement is not difficult to verify for $r=1$. In this case the $1$-graphons $U^{\alpha}$ and $V^\alpha$ can be regarded as indicator functions of measurable subsets $B^\alpha$ and $A^\alpha$ of $[0,1]$ (so for each $\alpha \in[k]$ we have $U^\alpha=\I_{B^\alpha}$ and $V^\alpha=\I_{a^\alpha}$) that form two partitions associated to $\UU$ and $\VV$ respectively. Note that $(A^\alpha)_{\alpha \in [k]}$ is obtained from $(B^\alpha)_{\alpha \in [k]}$ by the sampling process. A $k$-coloring corresponds to a refinement of these partitions with each original class being divided into $k$ measurable parts, that is $A^\alpha=\cup_{\beta \in [k]}^* A^{\alpha,\beta}$ and $V^{\alpha,\beta}=\I_{A^{\alpha,\beta}}$. Moreover, $|t(\FFF,\hat \UU)-t(\FFF,\hat \VV)|= |\prod_{l=1}^{q_0} \lambda(B^{\FFF(l)}) - \prod_{l=1}^{q_0}\lambda(A^{\FFF(l)})|$ for any of $k$-coloring $\hat \UU$ of $\UU$ and for any $[t] \times [k]$-colored $\FFF$ on $q_0$ vertices.  
We may define a suitable coloring by partitioning each of the sets $B^\alpha$ into parts $(B^{\alpha,\beta})_{\beta \in [k]}$ so that the classes satisfy $\lambda(B^{\alpha,\beta})=\lambda(B^{\alpha}) \frac{\lambda(A^{\alpha,\beta})}{\lambda(A^{\alpha})}$ when $\lambda(A^{\alpha}) > 0$, and $\lambda(B^{\alpha,\beta})=\lambda(B^{\alpha})\frac{1}{k}$ otherwise for each $\beta \in [k]$. Then by setting  $U^{\alpha,\beta}=\I_{B^{\alpha,\beta}}$ and  $\hat \UU=(U^{\alpha,\beta})_{\alpha \in [t],\beta \in [k]}$ we have that 
\begin{equation*} 
d_\tw (\mu(q_0, \hat \WW), \mu(q_0, \hat \UU)=  \frac{1}{2}\sum_{\FFF: |V(\FFF)|=q_0} |t(\FFF,\hat \UU)-t(\FFF,\hat \VV)| \leq \frac{ q_0^{k+1}}{2} \max_{\alpha \in [k]} |\lambda(A^{\alpha})-\lambda(B^{\alpha})|,
\end{equation*}
where the sum runs over all $[t] \times [k]$-colored $1$-graphs $\FFF$ on $q_0$ vertices.

The probability that for a fixed $\alpha \in [t]$  the deviation $|\lambda(A^\alpha)-\lambda(B^\alpha)|$ exceeds $\frac{2\delta }{q_0^{k+1}}$ is at most $2\exp(-\frac{4\delta^2  q}{3 q_0^{2k+2}})$, the union bound gives the upper bound $\exp(\ln 2+ t-\frac{4\delta^2  q}{3 q_0^{2k+2}})$ for the probability that $$  d_\tw (\mu(q_0, \hat \WW), \mu(q_0, \hat \UU)
\leq \delta$$
fails  for our particular choice for the coloring $\hat \UU$ of $\UU$. We note that the failure probability can be made arbitrary small with the right choice of $q$, so in particular smaller than $\delta$, therefore $q_\tw(1, \delta, q_0,t,k)= \frac{(t+\ln 2-\ln \delta) 3q_0^{2k+2}}{4\delta^2}$ that satisfies the conditions of the lemma. 

Now assume that we have already verified the statement of the lemma for $r-1$ and any other choice of the other parameters of $q_\tw$. Let us proceed to the proof of the case for $r$-graphons, therefore let   $\delta > 0$, $t,k,q_0 \geq 1$ be arbitrary and fixed, $q$ to be determined below and $\UU$, $\VV$, and $\hat \VV$ as in the condition of the lemma. We start by explicitly constructing a $k$-coloring $\hat \UU$ for $\UU$, in the second part of the proof we verify that the construction is suitable.

In a nutshell, we proceed as follows. We approximate $\hat \VV$ by the step function $\hat \ZZ$, and set $\ZZ=[\hat \ZZ,k]$, and also approximate $\UU$ by $\WW_1$. Let $\WW_2$ be the sampled version of $\WW_1$ generated by the same process as $\VV$. This way $\WW_2$ and $\ZZ$ are close, hence we can color $\WW_2$ using the coloring $\hat \ZZ$ of $\ZZ$ to obtain $\hat \WW_2$. The coloring $\hat \WW_2$ is then transferred onto $\WW_1$ using the induction hypothesis applied to the marginals of the step sets of $\WW_1$ and $\WW_2$ to get $\hat \WW_1$ with $[\hat \WW_1,k]=\WW_1$. Finally we color $\UU$ exploiting the proximity of $\UU$ and $\WW_1$ and the colored $\hat \WW_1$.

Our construction may fail to meet the criteria of the lemma, this can be caused at two points in the above outline. For one, it may happen, that $\WW_2$ does not approximate $\VV$ well enough, and secondly, when we transfer $\hat \WW_2$ onto $\WW_1$ using the induction hypothesis with $r-1$, as the current lemma leaves space for probabilistic error. These two events are independent from the particular choice of $\hat \VV$ and their probability can be made sufficiently small, we aim for to show this. We  proceed now to the technical details.

Let $\Delta= \frac{ \delta  r!}{4k (kt)^{q_0^r} q_0^r}$. Set $t_2=t_\reg(r,tk,\Delta, 1)$ and $t_1=t_\reg(r,t,\Delta/2, t_2)$, and define $q_\tw(r,\delta,q_0, t,k)=\max\{q_\tw(r-1,\delta/4,q_0,t_1,t_2), q_\cut(r,t,\Delta/2, t_1 t_2)\}$. Let $q \geq q_\tw(r,\delta,q_0, t,k)$ be arbitrary.  

First we apply \Cref{ch6:sregkhyper} with proximity parameter $\Delta/2$ to the $t$-colored $r$-graphon $\UU$, the lemma ensures the existence of a symmetric partition $\P=(P_1, \dots, P_{t_1})$ of $[0,1]^{\hhh([r-1])}$ with $t_\P \leq t_1$ 
 and a $t$-colored symmetric step function $\WW_1=(W_1^1, \dots, W_1^t)$ with steps in $\P$ that satisfies $\sup_{\QQ, t_\QQ \leq t_\P t_2}d_{\square,r,\QQ}(\WW_1,\UU) \leq \Delta/2, $ where the supremum runs over all symmetric partitions $\QQ$ of $[0,1]^{\hhh([r-1])}$ with at most $t_\P t_2$ classes. Applying structure preserving transformations to $[0,1]^{\hhh([r-1])}$ the classes of $\P$ can be considered as piled up, meaning  that for each $y \in [0,1]^{\hhh([r-1],r-2)}$ the fibers $\{y\} \times [0,1]$ are partitioned by the intersections with the classes of $\P$ into intervals $[0, a_1), [a_1, a_2),  \dots, [a_{t_1-1},a_{t_1}]$ with $\{y\} \times [a_{j-1}, a_j) = (\{y\} \times [0,1]) \cap P_j$. We introduce the  $r$-dimensional real arrays $A_1, \dots, A_t$ in order to describe the explicit form of the $W_1^\alpha$'s. So,
\begin{equation*}
W_1^\alpha(x_{\hhh([r],r-1)}) =  \sum_{i_1, \dots, i_r=1}^{t_\P} A_\alpha(i_1,\dots, i_r) \prod_{l=1}^r \I_{P_{i_l}} (x_{\hhh([r]\setminus \{l\})}) .
\end{equation*}

Let $\WW_2=(W^1_2, \dots, W^t_2)$ denote $\G(q,\WW_1)$, so $W_2^\alpha$ stands for $\G(q,W_1^\alpha)$ for each $\alpha \in [t]$, then \Cref{ch6:hyperregpres2} implies that for $q\geq q_\cut (r, t, \Delta/2, t_1t_2)$ it holds that 
\begin{equation*}
\sup_{\QQ, t_\QQ \leq t_\P t_2}d_{\square,r,\QQ}(\WW_2,\VV) \leq \Delta, 
\end{equation*}  
with probability at least $1-\Delta/2$, since $t_\P \leq r_1$. Also,
\begin{equation*}
W_2^\alpha(x_{\hhh([r],r-1)}) =  \sum_{i_1, \dots, i_r=1}^{t_\P} A_\alpha(i_1,\dots, i_r) \prod_{l=1}^r \I_{P'_{i_l}} (x_{\hhh([r]\setminus \{l\})}),
\end{equation*}
for each $\alpha \in [t]$ and $$P'_j=\cup_{(p_1, \dots, p_{r-1}) \in I_j} [\frac{p_1-1}{q},\frac{p_1}{q}] \times \dots \times  [\frac{p_r-1}{q},\frac{p_r}{q}] \times [0,1] \times \dots \times [0,1]$$  with $I_j=\{(p_1, \dots, p_{r-1}) : X_{r[\{p_1, \dots, p_{r-1}\}]} \in P_j \}$ for every $j \in [t_\P]$. Note that $\P'=(P'_j)_{j \in [t_\P]}$ is a symmetric partition. 

We apply now \Cref{ch6:sregkhyper} with proximity parameter $\Delta$ in order to approximate the $[t] \times [k]$-colored $r$-graph $\hat \VV=(V^{\alpha, \beta})_{\alpha \in [t], \beta \in [k]}$, the outcome is a $[t] \times [k]$-colored step function $\hat \ZZ=(Z^{\alpha, \beta})_{\alpha \in [t], \beta \in [k]}$ with symmetric steps in $\RRR=(R_1, \dots, R_{t_2})$ of $[0,1]^{\hhh([r-1]))}$ with $t_\RRR \leq t_2$ that satisfies
\begin{equation*}
 \sup_{\QQ, t_\QQ \leq t_\RRR} d_{\square,r,\QQ}(\hat \VV, \hat \ZZ) \leq \Delta.
\end{equation*} 
We introduce the $t$-colored step function $\ZZ=[\hat \ZZ,k]$  that is  the $k$-discoloring of $\hat \ZZ$ that has steps in $\RRR$ and note that \begin{equation*}
\sup_{\QQ,t_\QQ \leq t_\RRR} d_{\square,r,\QQ}(\VV, \ZZ) \leq \Delta, 
\end{equation*}
and therefore 
\begin{equation} \label{ch6:eq10}
\sup_{\QQ,t_\QQ \leq t_\RRR} d_{\square,r,\QQ}(\ZZ,  \WW_2) \leq 2\Delta. 
\end{equation}
Define the $r$-arrays $B_1, \dots, B_t$ such that for each $\alpha \in [t]$ it holds that \begin{equation*}
Z^\alpha(x_{\hhh([r],r-1)}) =  \sum_{i_1, \dots, i_r=1}^{t_\RRR} B_\alpha(i_1,\dots, i_r) \prod_{l=1}^r \I_{R_{i_l}} (x_{\hhh([r]\setminus \{l\})}),
\end{equation*}
further define also the $r$-arrays $(B_\alpha^\beta)_{\alpha \in [t], \beta \in [k]}$ such that 
\begin{equation*}
Z^{\alpha, \beta}(x_{\hhh([r],r-1)}) =  \sum_{i_1, \dots, i_r=1}^{t_\RRR} B^\beta_\alpha(i_1,\dots, i_r) \prod_{l=1}^r \I_{R_{i_l}} (x_{\hhh([r]\setminus \{l\})})
\end{equation*}
for each $\alpha \in [t], \beta \in [k]$. Clearly, $B_\alpha(i_1,\dots, i_r) =\sum_{\beta=1}^{k}B^\beta_\alpha(i_1,\dots, i_r)$.

Our aim next is to find a $k$-coloring of $\WW_2$ so that the new $tk$-colored $r$-graphon obtained is close to $\hat \ZZ$. In order to produce the coloring we apply \Cref{ch6:hypercoloring} relying on (\ref{ch6:eq10}), hence we obtain $\hat \WW_2$ with $[\hat \WW_2, k]=\WW_2$. The proximity between the two $tk$-colored $r$-graphs can be quantified by 
\begin{equation*}
\sup_{\QQ,t_\QQ \leq t_\RRR} d_{\square,r,\QQ}(\hat \ZZ, \hat \WW_2) \leq 2k\Delta.
\end{equation*}            
The graphon $\hat \WW_2$ is a symmetric step function with steps that form the coarsest partition that both refines $\P'$ and $\RRR$, we denote this symmetric partition of $[0,1]^{\hhh([r-1])}$ by $\Ss$, its number of classes satisfies $t_{\Ss} = t_{\P'} t_{\RRR} \leq t_1t_2$. 

Let us define the $t_\P$-colored $(r-1)$-graphon $\ww=(w^1, \dots, w^{t_\P})$ that is obtained from the classes of the partition $\P$ by integrating along the coordinate corresponding to the set $[r-1]$, that is $w^i(x_{\hhh([r-1],r-2)}) = \int_{0}^1 \I_{P_i} (x_{\hhh([r-1])}) \du x_{[r-1]}$. In the same way we define the $t_\P$-colored $(r-1)$-graphon $\uu=(u^1, \dots, u^{t_\P})$ corresponding to the partition $\P'$, as well as the $[t_\P] \times [t_\RRR]$-colored $\hat \uu=(u^{i,j})_{i \in [t_\P], j \in [t_\RRR]}$, where it holds that $\uu =[\hat \uu, t_\RRR]$ and $\hat \uu$ is the $t_\RRR$-coloring of $\uu$ corresponding to the partition $\Ss$. Note that $\ww$,$\uu$, and $\hat \uu$ satisfy the usual symmetries, since their origin partitions were symmetric. As the partition $\P'$ was constructed via the same sampling procedure as $\VV$ and $\WW_2$, therefore it holds that $\uu=\G(q,\ww)$ and $u^i=\G(q, w^i)$ for each $i \in [t_\P]$. 

We can assert that due to the induction hypothesis there exists a $t_\RRR$-coloring $\hat \ww=(w^{i,j})_{i \in [t_\P], j \in [t_\RRR]}$ of $\ww$ that satisfies 
\begin{equation*}
d_\tw (\mu(q_0, \ww), \mu(q_0, \uu) \leq \delta/4
\end{equation*}
with probability at least $1-\delta/4$ for $q \geq q_\tw(r-1,\delta/4, q_0,t_1,t_2)$. 

We construct a $k$-coloring for $\WW_1$ next. Recall the proof of \Cref{ch6:hypercoloring}, therefore we have that 

\begin{equation} \label{ch6:eq11}
W_2^{\alpha,\beta} (x_{\hhh([r],r-1)}) =  \sum_{i_1, \dots, i_r=1}^{t_\P} \sum_{j_1, \dots, j_r=1}^{t_\RRR} A_\alpha(i_1,\dots, i_r) \left[\frac{B^\beta_\alpha(j_1,\dots, j_r)}{B_\alpha(j_1,\dots, j_r)}\I_{B_\alpha>0} + \frac{1}{k}\I_{B_\alpha=0} \right]  \prod_{m=1}^r \I_{P'_{i_{m}} \cap R_{j_{m}}} (x_{\hhh([r]\setminus \{m\})}), 
\end{equation}
and set $A_\alpha^\beta((i_1,j_1), \dots, (i_r,j_r))=A_\alpha(i_1,\dots, i_r) \left[\frac{B^\beta_\alpha(j_1,\dots, j_r)}{B_\alpha(j_1,\dots, j_r)}\I_{B_\alpha>0} + \frac{1}{k}\I_{B_\alpha=0} \right]$ for all $\alpha \in [t]$, $\beta \in [k]$ and $((i_1,j_1), \dots, (i_r,j_r)) \in ([t_\P] \times [t_\RRR])^r$.

We utilize the $t_\RRR$-coloring $\hat \ww$ of the $(r-1)$-graphons $\ww$ to construct a refined partition of $\P$ that resembles $\Ss$ in order to enable the construction of a $k$-coloring of $\WW_1$ along the same lines as in (\ref{ch6:eq11}).
Let \begin{align*}& P_{i,j}=\{ x \in [0,1]^{\hhh([r-1])} : \\ &  \qquad \sum_{l=1}^{i-1} w^{l}(x_{\hhh([r-1],r-2)}) +\sum_{l=1}^{j-1} w^{i, l}(x_{\hhh([r-1],r-2)}) \leq x_{[r-1]} < \sum_{l=1}^{i-1} w^{l}(x_{\hhh([r-1],r-2)}) +\sum_{l=1}^{j} w^{i, l}(x_{\hhh([r-1],r-2)})  \}\end{align*} for each $i \in [t_\P]$ and $j \in [t_\RRR]$. Let $\P''=(P_{i,j})_{i \in [t_\P],j \in [t_\RRR]}$.

Clearly, $(P_{i,j})_{j \in [t_\RRR]}$ is a $t_\RRR$-partition of the set $P_i$, and $w^{i,j}(x_{\hhh([r-1],r-2)})=\int_{P_{i,j}} \du x_{\hhh([r-1])}$. We are able now to describe the $k$-coloring of the $\WW_1$, define
  
\begin{equation} \label{ch6:eq12}
W_1^{\alpha,\beta} (x_{\hhh([r],r-1)}) =  \sum_{i_1, \dots, i_r=1}^{t_\P} \sum_{j_1, \dots, j_r=1}^{t_\RRR} A_\alpha^\beta((i_1,j_1), \dots, (i_r,j_r)) \prod_{m=1}^r \I_{P_{i_{m},j_m}} (x_{\hhh([r]\setminus \{m\})}).
\end{equation}
Note that $\hat \WW_1$ is a step functions with a step partition that refines  $\P$ into $\P''$, but the regularity property of $\WW_1$ allows for 
\begin{equation*}
d_{\square,r,\P''}(\UU, \WW_1) \leq \Delta/2.
\end{equation*}
Finally we employ \Cref{ch6:hypercoloring} that produces a $k$-coloring $\hat \UU$ of $\UU$, so that $\hat \UU$ is $[t] \times [k]$-colored, and
\begin{equation*}
d_{\square,r,\P''}(\hat\UU,\hat \WW_1) \leq \frac{k\Delta}{2}.
\end{equation*} 
It remains to show that this coloring satisfies the requirements of the current lemma for a large enough $q$.

In the first step of the coloring construction we employed the $r$-graphon version of the intermediate regularity lemma, \Cref{ch6:sregkhyper}, therefore we can assert that for each $tk$-colored $\FFF$ we have by means of \Cref{ch6:hypercount} that 
\begin{align*}
\sum_{\FFF \in \GGG_{q_0}^{r,kt}} |t(\FFF,\hat \VV)-t(\FFF,\hat \ZZ)| \leq \frac{(kt)^{q_0^r} q_0^r}{r!} d_{\square,r}(\hat \VV, \hat \ZZ) \leq \frac{\delta}{4k},
\end{align*}
so we can conclude that $d_\tw (\mu(q_0,\hat \VV),\mu(q_0,\hat \ZZ)) \leq \frac{\delta}{8k}$. 

In the next step, as a consequence of \Cref{ch6:hypercoloring} and \Cref{ch6:hypercountcor}, we have for $\hat \WW_2$ that $d_\tw (\mu(q_0,\hat \ZZ),\mu(q_0,\hat \WW_2 )) \leq \delta/4$.  

We will next elaborate on the correctness of the inductive step of the construction.
Let us consider the $tk$-colored random $r$-graph $\G(q_0, \hat \WW_1)$, it is generated by the independent uniformly distributed $[0,1]$-valued random variables $\{Y_S: S \in \hhh([q_0], r) \}$. The color of each edge $e=\{e_1, \dots, e_2\} \in {[q_0] \choose r}$ is decided by determining first the unique tuple (up to coordinate permutations) $((i_1,j_1), \dots, (i_r,j_r)) \in ([t_\P]\times [t_\RRR])^{r}$ such that $(Y_{S})_{S \in \hhh(e \setminus \{e_l\})} \in P_{i_l,j_l}$, and then check for which pair $\alpha \in [t]$, $\beta \in [k]$ it holds that 
\begin{align*}
\sum_{l=1}^{\alpha-1} A_l((i_1,j_1), \dots, (i_r,j_r))&+ \sum_{l=1}^{\beta-1} A^l_\alpha((i_1,j_1), \dots, (i_r,j_r)) < Y_e \\ &\leq \sum_{l=1}^{\alpha-1} A_l((i_1,j_1), \dots, (i_r,j_r))+\sum_{l=1}^{\beta} A^l_\alpha((i_1,j_1), \dots, (i_r,j_r)), \end{align*} 
then add the color $(\alpha,\beta)$ to $e$ with corresponding index. It is convenient to view this process as first randomly $t_{\P''}$-coloring an  $(r-1)$-uniform template hypergraph $\GG_1$, whose edges are the simplices of the original edges, here we add a color $(i,j)$ to an $(r-1)$-edge $e'$ whenever $(Y_{S})_{S \in \hhh(e')} \in P_{i,j}$, and conditioned on $\GG_1$ we subsequently make independent choices for each edge to determine their color based on the arrays $A_\alpha^\beta$ by means of the random variables $\{Y_S: S \in {[q_0] \choose r} \}$ at the top level. 

Let us turn to the $tk$-colored  $\G(q_0,\hat \WW_2)$, the above description of the random process generating this object remains conceptually valid also for this random graph, the $r$-arrays $A_\alpha^\beta$ are identical to the case above, only the partition $\P''$ has to be altered to $\Ss$. Similarly as above, we introduce the random $(r-1)$-uniform $t_{\P''}$-colored hypergraph  $\GG_2$ that is generated as above adapted to $\G(q_0, \hat \WW_2)$. That means that the $(r-1)$-edges are colored by indices of the classes that form the partition $\Ss$ through the process that generates $\G(q_0, \hat\WW_2)$, see above. The key observation here is that conditioned on $\GG_1=\GG_2$, one can couple  $\G(q_0, \hat \WW_1)$ and $\G(q_0, \hat \WW_2)$ so that the two random graphs coincide with conditional probability $1$. Recall that a coupling is only another name for a joint probability space for the two random objects with the marginal distributions following $\mu(q_0, \WW_1)$ and $\mu(q_0,\WW_2)$ respectively. As the conditional distributions for the choices of colors for the $r$-edges are identical provided that $\GG_1=\GG_2$ the coupling is trivial. In order to construct a good unconditional coupling we require another coupling, now of $\GG_1$ and $\GG_2$, so that $\PPP(\GG_1\neq\GG_2)$ is negligible small for our purposes, and whose existence is exactly what the induction hypothesis ensures, when $q$ is large enough.   

As $q \geq q_\tw(r-1,\delta/4, q_0,t_1,t_2)$, the induction hypothesis enables us to use that there exist for any $\hat \uu$ a $\hat \ww$ so that $d_\tw(\mu(q_0,\hat \uu),\mu(q_0,\hat \ww )) \leq \delta/4$ holds with probability at least $1-\delta/4$ for each $\hat \uu$ simultaneously, which in turn implies that there is a coupling of the $t_1t_2$-colored random $(r-1)$-graphs $\GG_1$ and $\GG_2$ so that $\PPP(\GG_1 \neq \GG_2) \leq \delta/2$. 

It follows that there exists a coupling of $\G(q_0,\hat \WW_1)$ and $\G(q_0,\hat \WW_2)$  such that $\PPP(\G(q_0,\hat \WW_1) \neq \G(q_0,\hat \WW_2)) \leq \delta/2$ due to the discussion above, which in turn implies $$d_\tw (\mu(q_0,\hat \WW_1),\mu(q_0,\hat \WW_2 )) \leq \delta/4.$$ Since $\hat \WW_1$ has at most $t_\P t_2$ steps, another application of \Cref{ch6:hypercoloring} provides the bound $$d_\tw(\mu(q_0,\hat \WW),\mu(q_0,\hat \UU )) \leq \delta/16,$$
as $d_{\square,r}(\hat\UU,\hat \WW_1) \leq \frac{k\Delta}{2}.$

Evoking the triangle inequality and summing up the upper bounds on the respective deviations we conclude that 
\begin{align*}
d_\tw(\mu(q_0,&\hat \VV),\mu(q_0,\hat \UU)) \leq d_\tw(\mu(q_0,\hat \VV),\mu(q_0,\hat \ZZ))+ d_\tw(\mu(q_0,\hat \ZZ),\mu(q_0,\hat \WW_2))\\&+d_\tw(\mu(q_0,\hat \WW_2),\mu(q_0,\hat \WW_1))+d_\tw(\mu(q_0,\hat \WW_1),\mu(q_0,\hat \UU))
\leq \left(\frac{1}{8k}+\frac{1}{4}+\frac{1}{4}+\frac{1}{16}\right)\delta <\delta,
\end{align*} the overall error probability is at most $\delta/2+\Delta/2$, which is at most $\delta$.

\end{proof}

With \Cref{ch6:mainhyperlemma} at hand we can overcome the difficulties caused by properties of the $r$-cut norm for $r \geq 3$ in contrast to the case $r=2$, we turn to prove the main result of the paper.

\begin{proof}[Proof of \Cref{ch6:main}]
 We regard simple hypergraphs as $2$-colored $r$-graphs, in the following the term simple should be understood this way at each appearance. Let the $2k$-colored witness parameter of the nondeterministically testable $r$-graph parameter $f$ be denoted by $g$, whose sample complexity is at most $q_g(\varepsilon)$ for each proximity parameter  $\varepsilon > 0$. Set $d(r,\varepsilon, q_0,k,t)= \frac{[q_g(\varepsilon)^r\ln(tk)-\ln(\varepsilon)][2(tk)^{q_0^r}q_0^2]}{\varepsilon^2}$.Let $\varepsilon>0$ be fixed and define $q_f(\varepsilon)=\max\{q_\tw(r,\varepsilon/4, q_g(\varepsilon/4),k,3); \frac{4}{\varepsilon}q^2_g(\varepsilon/2);d(r,\varepsilon/4,q_g(\epsilon/4),k,2) \}$. We will show that for every $q \geq q_f(\varepsilon)$  the condition 
 \begin{align*}
 \PP(|f(G)-f(\G(q_f(\varepsilon),G)| >\varepsilon) < \varepsilon.
 \end{align*}
 is satisfied for each $G$. Let $q \geq q_f(\varepsilon)$ arbitrary but fixed and $G$ be a fixed simple graph on $n$ vertices.

First we show that $f(\G(q,G)) \geq f(G) - \varepsilon/4$ with probability at least $1-\varepsilon/4$. For this let us select a $k$-coloring $\GG$ of $G$ such that $f(G)=g(\GG)$, then the random $k$-colored graph $\FFF=\G(q,\GG)$ is a $k$-coloring of $\G(q,G)$, therefore $f(\G(q,G)) \geq g(\FFF)$, but since $q \geq q_q(\varepsilon/4)$ we know from the testability of $g$ that $g(\FFF) \geq g(\GG) - \varepsilon/4$ with probability at least $1-\varepsilon/4$, which verifies our claim.

The more difficult part is to show that $f(\G(q,G)) \leq f(G) + \varepsilon$ with failure probability at most $\varepsilon/2$. Let us  denote the random $r$-graph $\G(q,G)$ by $F$. We claim that with probability at least $1-\varepsilon/2$ there exists for any $k$-coloring $\FFF$ of $F$ there exists a $k$-coloring $\GG$ of G such that $|g(\FFF)-g(\GG)| \leq \varepsilon$, this suffices to verify the statement of the theorem.

Our proof exploits that the difference of the $g$ values between two colored $r$-graphs $\FFF$ and $\GG$ can be upper bounded by
\begin{align*}
|g(\FFF)-g(\GG)|  & \leq |g(\FFF) - g(\G(q_g(\varepsilon/4),\FFF))|+|g(\GG)-g(\G(q_g(\varepsilon/4),\GG))| \leq \varepsilon/2,
\end{align*}
whenever there exists a coupling of the two random $2k$-colored $r$-graphs $\G(q_g(\varepsilon/4),\GG)$ and $\G(q_g(\varepsilon/4),\FFF)$ appearing in the above formula such that they  are equal with probability larger than $\varepsilon/2$. Set $q_0=q_g(\varepsilon/4)$. We will show that with high probability fore every $\FFF$ there exists a $\GG$ that satisfies the previous conditions.

Recall that coupling is a probability space together with the random $r$-graphs $\GG_1$ and $\GG_2$ defined on it such that $\GG_1$ has the same marginal distribution as $\G(q_0,\GG)$ and $\GG_2$ has the same as $\G(q_0,\FFF)$, their joint distribution is constructed in a way that serves our current purposes by maximizing the probability that they coincide. When the target spaces are finite as in our case then a coupling that satisfies this condition can be easily constructed whenever 
$d_\tw (\mu(q_0,\GG), \mu(q_0,\FFF)) \leq 1 - \varepsilon/2$, see (\ref{ch6:eq101}). 

By \Cref{ch6:mainhyperlemma} for $3$-colored $r$-graphs (there are $3$ types of entries in the graphon representation of simple $r$-graphs, edges, non-edges, and diagonal elements) it follows that with probability at least  $1-\varepsilon/4$ for each $\FFF$  there exists a $3k$-colored $\UU$ with $[\UU,k]=W_G$ such that $d_\tw (\mu(q_0,\UU), \mu(q_0,\WW_\FFF)) \leq \varepsilon/4$. Let  us condition on this event and let $\FFF$ be fixed. From (\ref{ch6:eq10}) we know that $d_\tw (\mu(q_0, \GG), \mu(q_0, \WW_\GG)) \leq q_0^2/n\leq \varepsilon/4$ and $d_\tw (\mu(q_0, \FFF), \mu(q_0, \WW_\FFF)) \leq q_0^2/q \leq \varepsilon/4$. Further, it follows from our condition above that there exists a $3k$-colored $\UU$ that induces a fractional coloring of $G$, and $d_\tw (\mu(q_0,\UU), \mu(q_0,\WW_\FFF)) \leq \varepsilon/4$. It remains to produce a $3k$-coloring of $W_G$ from any fixed $3k$-colored $\UU$ ($k$ of the colors of $\UU$ correspond exclusively to diagonal cubes, so can be neglected). We do this by randomization, let $(X_{\{i\}})_{i \in [n]}$ be independent uniform random variables distributed on $[\frac{i-1}{n}, \frac{i}{n}]$, and let $(X_S)_{S \in \hhh([n],r)\setminus \hhh([n],1)}$ be independent uniform random variables on $[0,1]$. We can define $\WW_\GG$ to take the color $\UU(X_{\hhh(e)})$ on the set $ [\frac{e_1-1}{n}, \frac{e_1}{n}] \times \dots \times [\frac{e_r-1}{n}, \frac{e_r}{n}] \times [0,1] \times \dots \times [0,1]$ for $e=\{e_1, \dots, e_r\}$. For any fixed $\HHH \in \GGG_{q_0}^{r,2k}$ basic martingale methods deliver 
\begin{align*}
\PPP(|t(\HHH, \WW_\GG)-t(\HHH,\UU)| \geq \delta) \leq 2 \exp(-\frac{\delta^2 n}{2q_0^2})\end{align*}
 for any $\delta>0$, therefore when setting $\delta=\frac{\varepsilon}{4 (2k)^{q_0^r} }$ we get that
 \begin{align*}
 d_\tw (\mu(q_0,\UU), \mu(q_0,\WW_\GG)) = \frac{1}{2}\sum_{\HHH \in \GGG_{q_0}^{r,2k}}|t(\HHH, \WW_\GG)-t(\HHH,\UU)|\leq \varepsilon/4 
 \end{align*} 
 with probability at least $1-\varepsilon/4$, since $n \geq q \geq d(r,\varepsilon/4,q_0,k,2)= \frac{[q_0^r\ln(2k)-\ln(\varepsilon/4)][32(2k)^{q_0^r}q_0^2]}{\varepsilon^2}$. Summing up terms gives
 \begin{align*}
d_\tw(\mu(q_0,&\FFF),\mu(q_0,\GG)) \leq d_\tw(\mu(q_0,\FFF),\mu(q_0,\WW_\FFF))+ d_\tw(\mu(q_0, \WW_\FFF),\mu(q_0,\UU))\\&+d_\tw(\mu(q_0,\UU),\mu(q_0,\WW_\GG))+d_\tw(\mu(q_0,\WW_\GG),\mu(q_0,\GG))
\leq \varepsilon,
\end{align*}
with failure probability at most $\varepsilon/2$, this concludes our proof.

\end{proof}

\section{Nondeterministically testable hypergraph properties}\label{ch6:sec.prop}

The concept of nondeterministic testing was originally introduced for testing properties by \citet{LV}, and remarkable progress has been made in that context, see \cite{GS} and \cite{LV}, the estimation of parameters, which is our main issue in this paper, is in close relationship to that concept. For related developments in combinatorial property testing using regularity methods we refer to \cite{AFNS}.

 We present the definition of testability of properties in the usual and in the nondeterministic sense and construct a tester from the tester of the witness property with the aid of \Cref{ch6:mainhyperlemma} that achieves the same sample complexity as in the parameter testing case. This result connects our contribution to previous efforts more directly and answers the question posed in \cite{LV} asking if the equivalence of the two testability notions persists for uniform hypergraphs of higher order similar to the case of graphs.

\begin{definition}\label{ch6:def.proptest}
An $r$-graph property $\P$ is testable, if there exists another $r$-graph property $\hat \P$ called the sample property, such that
\begin{enumerate}[(a)]
\item $\PPP(\G(q,G)) \in \hat \P) \geq \frac{2}{3}$ for every $G \in \P$ and $q\geq 1$, and
\item for every $\varepsilon > 0$ there is an integer  $q_\P(\varepsilon)\geq 1$ such that for every $q \geq q_\P(\varepsilon)$ and every $G$ that is $\varepsilon$-far from $\P$ we have  that $\PPP(\G(q,G)) \in \hat \P) \leq \frac{1}{3}$.
\end{enumerate}
Testability for colored $r$-graphs is defined analogously.
\end{definition}

We remark that $\varepsilon$-far here means that one has to modify at least $\varepsilon|V(G)|^2$ edges in order to obtain an element of $\P$. Note that $\frac{2}{3}$ and $\frac{1}{3}$ in the definition can be replaced by arbitrary constants $1 > a >b >0$, this change may alter the corresponding certificate $\hat \P$ and the function $q_\P$, but not the characteristic of $\P$ being testable or not. Let $\P_n$ denote the elements of $\P$ of size $n$.

Next we formulate the definition of nondeterministic testability.

\begin{definition}
An $r$-graph property $\P$ is nondeterministically testable, if there exists an integer $k\geq 1$ and a $2k$-colored $r$-graph property $\QQ$ called the witness property that is testable in the sense of \Cref{ch6:def.proptest} satisfying $[\QQ,k]=\{[\GG,k] | \GG \in \QQ \}=\P$ (see \Cref{ch6:def.ndpartest} above for the discoloring operation).
\end{definition}

We formulate next the main theorem in this section.
\begin{theorem}
Every nondeterministically testable $r$-graph property is testable.
\end{theorem}

\begin{proof}
Let $\P$ be a nondeterministically testable property with witness property $\QQ$ of $2k$-colored $r$-graphs. We employ the combinatorial language with counting subgraph densities when referring to $\QQ$ and its testability, and the probabilistic language of picking random subgraphs in a uniform way when handling $\P$ in order to facilitate readability.  

Let $\hat\QQ$ be the corresponding sample property that certifies the testability of $\QQ$ and $q_\QQ$ be the sample complexity corresponding to the thresholds $1/5$ and $4/5$, that is
\begin{enumerate}[(i)]
\item if $\GG \in \QQ$, then for every and $q\geq 1$ we have $t(\hat \QQ_q, \GG) \geq 4/5$, and
\item for every $\varepsilon > 0$ if $\GG$ is $\varepsilon$-far from $\QQ$, then for every $q \geq q_\QQ(\varepsilon)$ we have that $t(\hat \QQ_q, \GG) \leq 1/5$.
\end{enumerate}
Our task is to construct a property $\hat \P$ together with a function $q_\P$ such that they fulfill the conditions of \Cref{ch6:def.proptest}. We are free to specify the error thresholds by the remark after \Cref{ch6:def.proptest}, we set them to $2/5$ and $3/5$. 

Let $n$ be a positive integer and let $\varepsilon_n >0$ be the infimum of all positive reals $\delta$ that satisfy $n \geq \max\{q_\tw(r,1/10,q_Q(\delta),3,k);100 q^2_\QQ(\delta);d(r,1/10,q_\QQ(\delta),k,2)\}$ from \Cref{ch6:mainhyperlemma}.
Define for each $n$ the set 
\begin{align*}
\hat \P_n  = \{H \in \GGG_n^r | \textrm{there exists a $k$-coloring $\HHH$ of $H$ such that } t(\hat Q_{q_Q(\varepsilon_n)},\HHH) \geq 3/5\},
\end{align*}
and let $\hat \P = \cup_{n=1}^\infty \hat \P_n$. We set $q_\P(\varepsilon) = \max\{q_\tw(r,1/10,q_Q(\varepsilon),2,k);100q^2_\QQ(\varepsilon);d(r,1/10,q_\QQ(\delta),k,2) \}$. We are left to check if the two conditions for testability of $\P$ hold with $\hat \P$ and $q_\P$ described as above. Assume for the rest of the proof that $n \geq q_\P(\varepsilon_n)$ for each $n$ for simplicity, the general case follows along the same lines with some technical difficulties.

First let $G \in \P$, we have to show that for every $q \geq 1$ integer we have that $\G(q, G)\in \hat \P_q$ with probability at least $3/5$. 

The condition $G \in \P$ implies that there exists a a $k$-coloring $\GG$ of $G$ such that $\GG \in \QQ$. From the testability of $\QQ$ it follows that $t(\hat \QQ_l, \G(l, \GG)) \geq 4/5$ for any $l \geq 1$. Let $q \geq 1$ be arbitrary, and let $F$ denote $\G(q,G)$, furthermore let $\FFF=\G(q,\GG)$ generated by the same random process as $F$, so $\FFF$ is  a $k$-coloring of $F$. We know by a standard sampling argument that
\begin{align}\label{ch6:eq62}
\PPP(| t(\hat \QQ_{q_\QQ(\varepsilon_q)},\GG)-t(\hat \QQ_{q_\QQ(\varepsilon_q)},\FFF)| \geq 1/5) \leq 2 \exp\left(- \frac{q}{50 q^2_\QQ(\varepsilon_q)} \right),
\end{align}
and the right hand side of (\ref{ch6:eq62}) is less than $2/5$ by the choice of $\varepsilon_q$, since by definition $q \geq 100q^2_\QQ(\varepsilon_q)$. It follows that $t(\hat \QQ_{q_\QQ(\varepsilon_q)},\FFF) \geq 3/5$ with probability at least $3/5$, so by the definition of $\hat P$ we have that $F \in \hat \P_q$ with probability at least $3/5$, which is what we wanted to show.

To verify the second condition we proceed by contradiction. Suppose that $G$ is $\varepsilon$-far from $\P$, but at the same time  there exists an $l \geq q_\P(\varepsilon)$ such  that $F \in \hat \P_l$ with probability larger than $2/5$, where $F=\G(l, G)$.

 In this case, the latter condition implies that with probability larger than $2/5$  there exists a $k$-coloring $\FFF$ of $F$  such that  $t(\hat \QQ_{q_\QQ(\varepsilon_{l})},\FFF) \geq 3/5$. By \Cref{ch6:mainhyperlemma} and the proof of \Cref{ch6:main} there exists a $k$-coloring $\GG$ of $G$ such that $d_\tw(\mu(q_\QQ(\varepsilon_{l}),\FFF),\mu(q_\QQ(\varepsilon_{l}),\GG)) \leq 22/100$ with probability at least $4/5$, in particular
\begin{align*}
|t(\hat \QQ_{q_\QQ(\varepsilon_{l})},\FFF) - t(\hat \QQ_{q_\QQ(\varepsilon_{l})},\GG) | \leq \frac{22}{100},
\end{align*}
which implies that with probability at least $1/5$ there exist a $\GG$ such that  $t(\hat \QQ_{q_\QQ(\varepsilon_{l})},\GG) > \frac{3}{10}$. 
We can drop the probabilistic assertion and can say that there exists a $k$-coloring $\GG$ of $G$ such that $t(\hat \QQ_{q_Q(\varepsilon_{l})},\GG) > \frac{3}{10}$, because $G$ and the density expression are deterministic.  

On the other hand, the fact that $G$ is $\varepsilon$-far from $\P$ implies that any $k$-coloring $\GG$ of $G$ is $\varepsilon$-far from $\QQ$, which means that $t(\hat \QQ_{q},\GG) \leq 1/5$ for any $k$-coloring $\GG$ of $G$ and $q \geq q_Q(\varepsilon)$. But we know that $q_Q(\varepsilon_{l}) \geq q_Q(\varepsilon)$, since $\varepsilon_l \leq \varepsilon$  which delivers the contradiction. The last inequality is the consequence of our definitions, $\varepsilon_l$ is the infimum of the $\delta>0$ that satisfy $l \geq  q_\P(\delta)$, and on the other hand, $l \geq  q_\P(\varepsilon)$.

\end{proof}

\section{Further research}\label{ch6:sec.further}

It would be very interesting to shed light on the explicit sample complexity  bounds for the witness parameter in \Cref{ch6:main}. The only ingredient of our proof which is non-effective is the part which deals with the ultralimit method in the proof of \Cref{ch5:samp},  to our knowledge an effective proof regarding this result is only known for $r=2$.

\subsection*{Acknowledgement}
We thank G\'abor Elek for a number of interesting discussions.

\bibliographystyle{notplainnat}
\bibliography{thesis_refs}

\end{document}